\documentclass[aps,pra,showpacs,twoside,10pt,floatfix,nofootinbib,longbibliography,twocolumn]{revtex4-2}
\usepackage{times,epsfig,amssymb,amsfonts,amsmath,bm,subfigure,mathtools,amsthm,braket,soul,enumitem,xcolor, physics, graphics,graphicx}
\usepackage[normalem]{ulem}
\usepackage{comment, appendix}
\usepackage{epstopdf}
\newtheorem{theorem}{Theorem}

\usepackage{comment}
\usepackage[colorlinks=true, citecolor=red, urlcolor=blue ]{hyperref}

\begin{document}

\title{
Dynamically frozen long-distance entanglement via non-Hermitian \(\hat{\mathcal{P}}\hat{\mathcal{T}}\)-symmetric systems}

\author{Sejal Ahuja$^{1,2}$, Keshav Das Agarwal$^{1,2}$,  Aditi Sen(De)$^{1,2}$}
\affiliation{$^1$Harish-Chandra Research Institute,  Chhatnag Road, Jhunsi, Prayagraj - 211019, India}
\affiliation{$^2$Homi Bhabha National Institute,  Training School Complex, Anushakti Nagar, Mumbai 400 094, India}

\begin{abstract}

In distributed quantum networks, interacting spin systems can mediate the generation of highly entangled links between distant nodes. We investigate the role of effective parity-time ($\hat{\mathcal{P}}\hat{\mathcal{T}}$)-symmetric non-Hermitian spin-$1/2$ bulks weakly coupled to two quantum links, obtained due to the environmental interactions affecting both the bulk and the links. Focusing on effective non-Hermitian nearest-neighbor (NN) Su-Schrieffer-Heeger (SSH) models, we analyze how non-Hermiticity influences the dynamical formation of long-distance entanglement (LDE). For a paradigmatic model consisting of a quantum $XX$ bulk subjected to imaginary staggered magnetic fields, we analytically determine the exceptional points arising from the resulting bulk-mediated interactions between the links. Combining analytical and numerical methods, we demonstrate that an initially fully separable state can dynamically evolve into highly entangled link states near these exceptional points in the broken regime. Further, after optimizing over time and system parameters, near-unit time-averaged entanglement between the links emerges under weak imaginary magnetic fields and bulk-link couplings, which  cannot be attained in the corresponding Hermitian systems. Moreover, the non-Hermitian dynamics exhibit a {\it freezing} of high entanglement in the vicinity of exceptional points, a feature absent in Hermitian counterparts. We also identify regimes of long-range interaction strengths that yield a higher time-averaged entanglement than the corresponding NN models. Furthermore, we establish that LDE persists in the stationary regime, highlighting the promise of engineered non-Hermitian dynamics for realizing robust and frozen entangled links in quantum networks.

\end{abstract}

\maketitle

\section{Introduction}

Numerous counterintuitive and exotic phenomena have emerged with the development of quantum physics, among which quantum entanglement occupies a central position \cite{Schrodinger_1935, Einstein_1935, Bell_1964, Aspect_1982, HoroRMP, nielsenbook}. Beyond its foundational significance, it has become one of the key resources driving the second quantum revolution, enabling the realization of a wide range of quantum technologies~\cite{benet_tele,nielsenbook, HoroRMP, Chitambar_rmp}, including quantum communication devices and computers~\cite{Bennett2000, *Gisin2007, *Northup2014}, and quantum sensors~\cite{degan_review,QmetroRMP}. In particular, entanglement shared between spatially separated quantum systems is required to achieve scalable quantum networks~\cite{Cirac_PRA_1999, *repeater_PRL_1998, *Repeater_PRA_1999, *Duan2001, *Kimble2008, *Halder_2022, *Azuma_2023, Meter2011}. Such long-distance entangled links constitute the building blocks for connecting distant quantum processors, thereby providing a framework for modular architectures~\cite{Kielpinski2002,*Blinov2004,*Skinner_PRL_2003} and distributed quantum computation~\cite{Meter2011,Caleffi_2024}. In this context, one-dimensional quantum spin systems have attracted significant attention as promising platforms for the generation and preservation of long-distance entanglement (LDE)~\cite{Bose, Bose_2003, Wojcik05, venuti_prl, VenutiPRA07, Ferreira_PRA_2008, Giampaolo_2010, Dhar_2016, Trifunovic_2013, Abaach2023, Soares_2025, agarwal2026_b} between distant spins interacting weakly with the bulk (due to monogamy of entanglement~\cite{wooters_monogamy_pra_2000}) (see Fig.~\ref{fig:schematics} for schematic representation of the set-up). This architecture can create highly entangled quantum links mediated by the bulk, through both equilibrium ~\cite{VenutiPRA07, venuti_prl, Ferreira_PRA_2008, Giampaolo_2010, Dhar_2016, agarwal2026_b} and dynamical scenarios~\cite{Bose, Bose_2003, Wojcik05, Trifunovic_2013, Abaach2023, Soares_2025, agarwal2026_b}.

\begin{figure}
    \centering
\includegraphics[width=\linewidth]{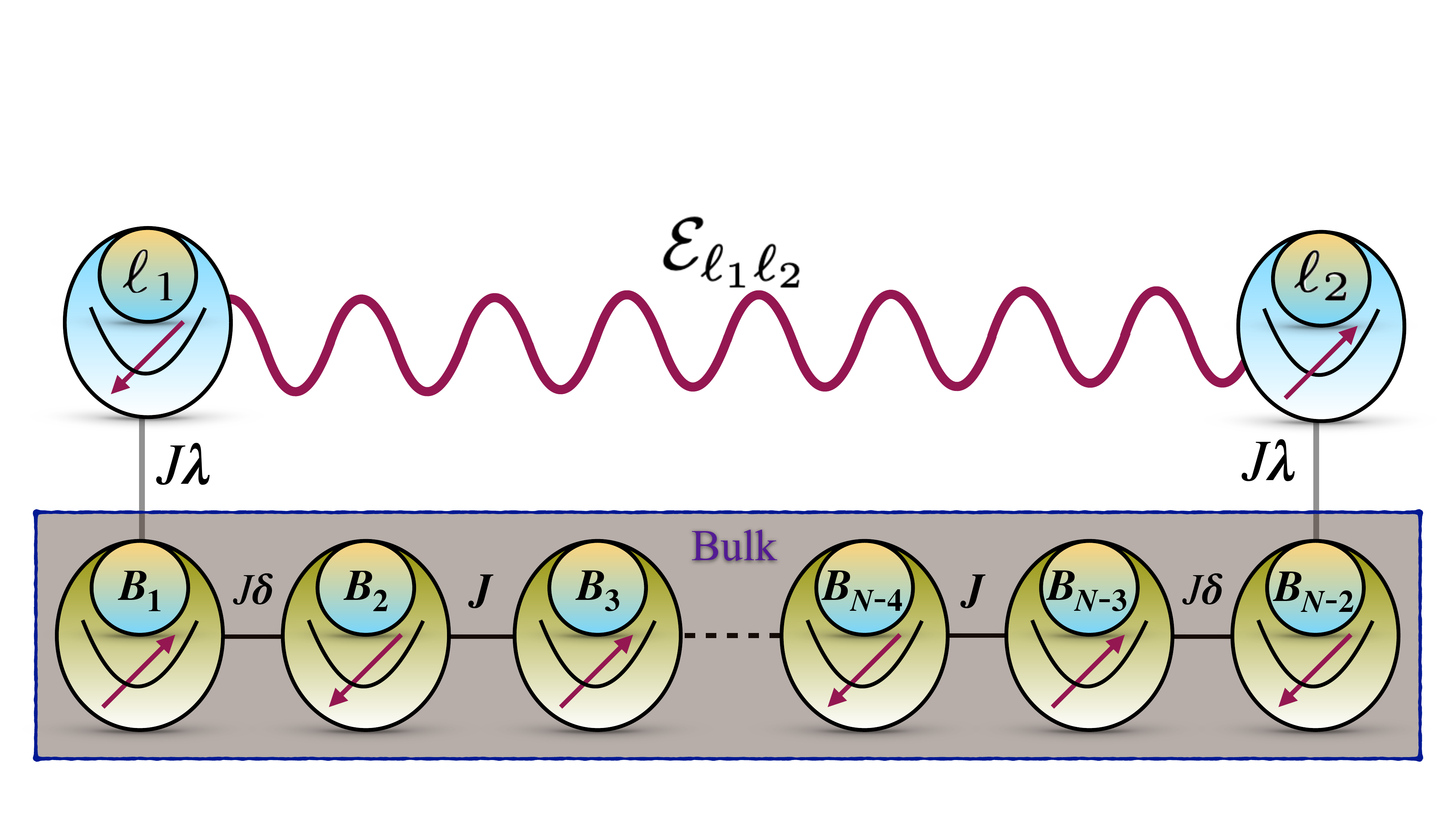}
    \caption{{\bf Schematic of the LDE generation protocol when all the sites are influenced by the environment.} The spin chain is composed of two components: bulk \((B_1, B_2, \cdots, B_{N-2})\) interacting weakly with the two link spins \((\ell_1,\ell_N)\) situated at the ends of the bulk. In our case, the interactions in bulk spins are described by the SSH model (see Eq. (\ref{eq:ssh_Ham_def})), and the weak interactions between link-bulk spin pairs \((\ell_1,B_1)\) and \((B_{N-2},\ell_2)\) are represented by \(\lambda\ll1\).  Each spin is attached to a local bath, where continuous measurements are performed in opposite directions to obtain an effective Hamiltonian of an imaginary alternating magnetic field. The protocol aims to create near-maximal long-distance entanglement between the links \(\mathcal{E}_{\ell_1 \ell_2}\sim 1\).}
    \label{fig:schematics}
\end{figure}


One of the current challenges confronting the advancement of quantum technologies is the unavoidable interactions with the surrounding environment, leading to the degradation of quantum resources, especially entanglement, over time. Understanding how to mitigate these detrimental effects, engineer robust quantum devices, or identify conditions in which quantum protocols retain their efficiency despite environmental influences is therefore of paramount importance. In recent years, non-Hermitian systems have emerged as an effective theoretical framework for describing and exploiting the effects of system-reservoir coupling. In particular, open quantum systems governed by the Gorini–Kossakowski–Lindblad–Sudarshan (GKLS) master equation~\cite{open_quan_book, lidar_2020_lecture} can,  through continuous monitoring and post-selected dynamics, be mapped to an effective closed system description generated by a non-Hermitian Hamiltonian~\cite{Eleuch2015, Minganti2020, turkeshi_prb_2021, turkeshi_prb_2023b}. The resulting dynamics are non-unitary and are governed by an effective Hamiltonian that incorporates the dissipation and measurements ~\cite{Ashida2020}. Such non-Hermitian systems can possess anti-unitary symmetries like the parity-time- (\(\hat{\mathcal{P}}\hat{\mathcal{T}}\)-)~\cite{bender_spectra, mosta_ali_2002, Mostafazadeh_2004} and rotation-time-symmetry~\cite{Song_RT_symm, gan_adi_factor, kda2024, Agarwal2026}, which gives rise to several intriguing phenomena.  Prominent ones include the emergence of unbroken and broken phases separated by exceptional points~\cite{bender_spectra, Minganti2020, prl_exp_EPnonmarkov}, non-trivial entanglement properties~\cite{turkeshi_prb_2023b, Kawabata2023, Agarwal2026}, the non-Hermitian skin effect~\cite{Li2020Oct, xiao_prb_2022} and exceptional topology~\cite{fleckenstein_prr_2022, *Bergholtz_rmp_2021, *Kawabata2019}, all of which are distinctive characteristics of non-Hermitian systems with no direct Hermitian analogue. 
Recent experimental realizations across photonic, atomic, and other engineered platforms~\cite{guo_pt_exp,*chitsazi_prl_2017,*naghiloo_nature_ep_sensing_2019,*ep_sensing_exp_2} have demonstrated the practical relevance of non-Hermitian phenomena and their potential advantages for emerging quantum technologies
such as quantum sensing~\cite{ep_sensing_3, *ep_sensing_2, *ep_sensing_1, *Budich_2020, *Xiao_2024, *Parto2025, *agarwal2025_sensing, *agarwal2026_sensing, *Wiersig2026} and quantum thermal machines~\cite{khandelwal_prx_2021, *konar2022quantum, *Santos2023Dec}, thereby opening new avenues for environmental interactions as resources.

Despite significant progress in understanding entanglement generation in isolated systems, the influence of environmental interactions on quantum networks, particularly when both the processors and the connecting links are coupled to external baths, remains largely unexplored. In this work, we aim to bridge this gap by asking the question: "Can long-distance entanglement be dynamically created between distant links when the underlying effective dynamics is governed by a non-Hermitian Hamiltonian?'' We answer this question affirmatively by demonstrating that near-maximal and remarkably stable entanglement can emerge dynamically from an initially fully separable state.

To establish these results, we consider a paradigmatic parity-time \(\hat{\mathcal{P}}\hat{\mathcal{T}}\)-symmetric Su-Schrieffer-Heeger (SSH) model~\cite{ssh_og, Asboth2016, Lieu2018, Halder_2023, Han2021}, which has been extensively studied in the contexts of topological properties~\cite{Lieu2018, Halder_2023, Han2021}, half-block-entanglement dynamics ~\cite{shiro_scip_25, xhek_scip_2023, adhip_PRB22, noclck_xhek, fazio_prb25, nhent_prb} and entanglement revival via non-Markovian dynamics \cite{soo_pra25}. Such a non-Hermitian class of models has also been experimentally realized in platforms, such as superconducting circuits \cite{Mei2018_pra} and photonic waveguides \cite{Luo2024_aqt}, making them potential candidates for investigating LDE production. Motivated by recent investigations of entangled-link generation in both closed~\cite{Lang2017, zhang_25} and open~\cite{jaynes_nh_aop,hamid_ep,prl_25_condent, Soares_2025,ent_dyn_atom2026} systems, we explore the role of non-unitary evolution in facilitating robust LDE creation (see \cite{pra_feedback} for feedback with bosonic bulk and \cite{ent_ep1,nori_prr_entEP} for experimental validation of enhanced entanglement near exceptional points). 

Towards achieving this objective, we analytically determine the exceptional points of an $XX$ spin chain with nearest-neighbor (NN) interactions subjected to an imaginary staggered magnetic field in the weak coupling regime between the end links and the bulk. We show that, under the evolution through the \(\hat{\mathcal{P}}\hat{\mathcal{T}}\)-symmetric SSH model, entanglement between distant links exhibits a striking freezing of entanglement near the exceptional point in the broken phase, maintaining almost maximal values over extended periods of time.  In the (XX)-limit of the SSH model, the generation of maximally entangled links can be established analytically, while for other classes of the SSH models, numerical analyses reveal that the time-averaged entanglement is almost unity in the proximity of exceptional points. Furthermore, by comparing the long-distance entanglement generated via non-Hermitian and corresponding Hermitian evolutions, we find that the non-Hermitian dynamics can produce higher entanglement after optimizing over the system parameters and evolution times. Remarkably, this enhancement is robust against variations in the initial state. Our analysis identifies the non-Hermitian ($XX$) model as the optimal setting for entangled-link generation within the considered family of models. Extending beyond NN interactions, we demonstrate that appropriately tuned long-range couplings can further enhance the generated average optimized long-distance entanglement compared with their nearest-neighbor counterparts.

We further extend our analysis beyond the non-equilibrium regime to the stationary states governing the long-time dynamics (\(t \rightarrow \infty\)) of the system. We find that these non-Hermitian stationary states can support substantial LDE between links weakly coupled to the bulk. This achievable LDE exceeds that present in the ground state of the corresponding Hermitian counterpart, highlighting the constructive role of non-Hermiticity in generating robust quantum correlations.


The paper is organized as follows. Sec. \ref{sec:framework} introduces the protocol of long-distance entanglement generation, the non-Hermitian \(\hat{\mathcal{P}}\hat{\mathcal{T}}\)-symmetric interactive model and the figure of merits for our work. We analytically demonstrate the expression for the exceptional points for the non-Hermitian SSH model in Sec. \ref{sec:ep_nH}. In Sec. \ref{sec:nH_dyn}, we present the benefit of non-Hermiticity in LDE creation, and its long-time freezing in this setting, whereas the study in the stationary case is presented in Sec. \ref{sec:statics}. Sec. \ref{sec:conclu} summarizes the results.

\section{Framework for generating entangled links}
\label{sec:framework}

In order to create long-distance entanglement between distant parties (links), interactions between them are necessary. These interactions can be either direct or mediated via the channel (bulk) connecting them~\cite{venuti_prl, VenutiPRA07, Ferreira_PRA_2008, Giampaolo_2010, Dhar_2016, Trifunovic_2013, Abaach2023, Soares_2025, agarwal2026_b} which we will adopt in this work. Let us consider a chain of \(N\) spin-\(1/2\) parties, with quantum links (\(\ell_1\) and \(\ell_2\)) being the spins at the ends of the chain, i.e., at sites $(\ell_1, \ell_2)\equiv(1, N)$ and the bulk (intermediate) sites being represented by \(B_k\equiv k-1\) for $k=2$ to $N\!-\!1$ as shown in Fig.~\ref{fig:schematics}. Further, a key distinction of our approach from previous protocols connecting sites \(1\) and \(N\) is that all the spins in the chain are coupled to the bath, giving rise to a novel class of non-Hermitian models.

We first present the protocol for the dynamical generation of highly entangled quantum links via $U(1)$-symmetric Hamiltonian. In particular, we investigate the \(\hat{\mathcal{P}}\hat{\mathcal{T}}\)-symmetric SSH model~\cite{ssh_og, Asboth2016, Lieu2018, Halder_2023, Han2021} for building entangled links $(\ell_1, \ell_2)$ providing LDE. The bulk with open boundaries is represented by the Hamiltonian \(\hat{H}_B\) and the end spins of the bulk are coupled to the links via $\hat{V}_{\ell}$ for even $N$. Therefore, the Hamiltonian $\hat{\mathcal{H}}_{B\ell}$ of the system is given by
\begin{align}
    \nonumber \hat{\mathcal{H}}_{B\ell} &= J\hat{H}_B + \lambda^\prime \hat{V}_{\ell}, \quad \hat{H}_{B} = \! \delta \sum_{\mathclap{\substack{k=2\\k\in  \text{even}}}}^{\mathclap{N-2}}\hat{\mathsf{H}}_{k, k+1} + \sum_{\mathclap{\substack{k=3\\k\in  \text{odd}}}}^{\mathclap{N-3}}\hat{\mathsf{H}}_{k, k+1}, \\  
    \hat{V}_{\ell} &= \hat{\mathsf{H}}_{1, 2} + \hat{\mathsf{H}}_{N-1, N}, \quad \hat{\mathsf{H}}_{a, b} = \hat{\sigma}^x_{a}\hat{\sigma}^x_{b} \!+\! \hat{\sigma}^y_{a}\hat{\sigma}^y_{b} ,
    \label{eq:ssh_Ham_def}
\end{align}
where \(\hat{\sigma}_k^p\) ($p\!=\!x,y,z$) are the Pauli matrices at site \(k\) with $\hat{\sigma}_k^z = \ket{\uparrow}\!\bra{\uparrow}_k \!-\! \ket{\downarrow}\!\bra{\downarrow}_k$, \(\delta\) denotes the intracell couplings relative to the intercell couplings and \(\lambda\equiv\lambda^\prime/J\) corresponds to the bulk-link coupling strength at the ends of the spin chain. The $XX$ model is given by the uniform interactions in the bulk, i.e., \(\delta \!=\! 1\), which has been studied for the LDE present in the ground state~\cite{VenutiPRA07, Giampaolo_2010, Dhar_2016},  as well as for quantum teleportation~\cite{venuti_prl_2007, VenutiPRA07}.

Let us now assume that all the \(N\) sites interact with environmental qubits (as shown schematically in Fig. \ref{fig:schematics}) and hence the joint system-environment duo undergoes the evolution together. After continuously measuring the auxiliary qubit on the eigenbasis of \(\hat{\sigma}^z\) and post selecting the outcomes, the effective dynamics of the system is governed by the non-Hermitian Hamiltonian having imaginary alternative magnetic field at every site (see Appendix~\ref{app:nH_deriv} for details)~\cite{Minganti2020, turkeshi_prb_2021, turkeshi_prb_2023b}. In particular, in our setup described in Eq. (\ref{eq:ssh_Ham_def}), the \(\hat{\mathcal{P}}\hat{\mathcal{T}}\)-symmetric non-Hermitian SSH model in the presence staggered (alternating) imaginary magnetic field described by \(\hat{W}\), and its Hermitian counterpart with real magnetic field, are given by 
\begin{align}
    \label{eq:nhH_def}
    &\hat{\mathcal{H}}_{\text{nH}} = \hat{\mathcal{H}}_{B\ell} + ih^\prime\hat{W}, \quad \hat{\mathcal{H}}_{\text{H}} = \hat{\mathcal{H}}_{B\ell} + h^\prime\hat{W}, \\
    \nonumber &\hat{W} = \sum_{k=1}^N (-1)^{k+1} \hat{\sigma}_k^z ,
\end{align}
where $h\equiv h^\prime/J$ is the strength of the magnetic field. We show the advantage of LDE generation via $\hat{\mathcal{H}}_{\text{nH}}$ over its Hermitian counterpart $\hat{\mathcal{H}}_{\text{H}}$, given by $h\leftrightarrow ih$. While the correspondence between the topological properties between the Hermitian and non-Hermitian systems has been explored~\cite{Lee_2019, *Schindler_2023, *Hamanaka_2024}, the $h\leftrightarrow ih$ correspondence has been recently shown for various properties, namely between the factorization surface of the Hermitian model with the exceptional points of the corresponding non-Hermitian system~\cite{gan_adi_factor, Agarwal2026, ahuja_nhqst},  as well as between system parameters with unsuccessful quantum state transfer~\cite{ahuja_nhqst}. 

The SSH model possess the $U(1)$-symmetry given by the operator $\hat{n}^{(\uparrow)}\!\!=\!\!\sum_{k=1}^N\ket{\uparrow}\!\bra{\uparrow}_k$ with $[\hat{n}^{(\uparrow)}, \hat{\mathcal{H}}_{\text{X}=\text{nH}, \text{H}}] \!=\! 0$, and its eigenspace with eigenvalues $n^{(\uparrow)} \!=\! 0,1,\dots N$ denotes the $N\!+\!1$ different sectors in the computational basis $\{\ket{\uparrow}, \ket{\downarrow}\}^{\otimes N}$. 
In our work, we consider the dynamics in the $n^{(\uparrow)} \!=\! 1$ sector, which gives the $N$-dimensional Hilbert space, with non-trivial dynamics, arising from $\{\ket{\underline{k}}=\hat{\sigma}^x_k \ket{\downarrow}^{\otimes N}\}_{k=1}^{N}$ basis, where $\ket{\underline{k}}$ denotes the state, with $\ket{\uparrow}$ being only at $k$-th site. 
Therefore, the interactions $\hat{\mathsf{H}}_{a, b}$ can be represented as $\hat{\mathsf{H}}_{a,b} \!=\! 2\left(\ket{\underline{a}}\!\bra{\underline{b}} \!+\! \ket{\underline{b}}\!\bra{\underline{a}}\right)$ in the $n^{(\uparrow)} \!=\! 1$ sector basis and describes the hopping of non-interacting particles in a dimerized chain~\cite{ssh_og, Asboth2016}, while $\hat{W} \!=\! 2\sum_{k=1}^{N} (-1)^{k+1}\ket{\underline{k}}\!\bra{\underline{k}}$ is diagonal and provides a staggered potential at each site.

In addition to the \(U(1)\) symmetry, the SSH model also possesses \(\hat{\mathcal{P}}\hat{\mathcal{T}}\)-symmetry~\cite{bender_spectra, mosta_ali_2002, Mostafazadeh_2004}, with \(\hat{\mathcal{P}}\) being the parity operator and \(\hat{\mathcal{T}}\) the time reversal operator. The parity symmetry corresponds to $\hat{\mathcal{P}}\hat{A}_k \hat{\mathcal{P}} \!=\! \hat{A}_{N-k+1}$ with an arbitrary local operator \(\hat{A}_k\) at the site $k$. The time reversal symmetry is given by
$\hat{\mathcal{T}}\hat{A}\hat{\mathcal{T}} \!=\! \hat{A}^{*}$ for an arbitrary \(\hat{A}\). While the Hamiltonian $\hat{\mathcal{H}}_{\text{H}}$ preserves both the \(\hat{\mathcal{P}}\) and \(\hat{\mathcal{T}}\) symmetries individually, the non-Hermitian model $\hat{\mathcal{H}}_{\text{nH}}$ breaks them and preserves the overall \(\hat{\mathcal{P}}\hat{\mathcal{T}}\)-symmetry. When the eigenspectrum of the system $\hat{\mathcal{H}}_{\text{nH}}$ preserves the \(\hat{\mathcal{P}}\hat{\mathcal{T}}\)-symmetry, the eigenvalues are real and the non-Hermitian Hamiltonian can be mapped to a Hermitian Hamiltonian via a similarity transformation~\cite{mosta_ali_2002, Mostafazadeh_2004}, hence also called pseudo-Hermitian Hamiltonian, and the system lies in the unbroken phase. In the broken phase, eigenspectrum can break the \(\hat{\mathcal{P}}\hat{\mathcal{T}}\)-symmetry and the eigenvalues become imaginary. The system possesses exceptional points where this transition from the unbroken to the broken regime occurs and eigenvectors coalesce as the Hamiltonian becomes defective.

In order to study the dynamics of the spin chain, the bulk is initialized in a product eigenstate of bulk system $\hat{H}_B$, specifically as $\ket{\psi}_B=\otimes_{k=2}^{N-1}\ket{\downarrow}_k$ while links are initialized as $\ket{\uparrow}_{\ell_1}$ and $\ket{\downarrow}_{\ell_2}$. Therefore, the initial state is a fully separable state 
\begin{equation}
    \ket{\Psi(0)} =\ket{\uparrow}_{\ell_1} \otimes \ket{\downarrow}_{B,\ell_2}^{\otimes N-1} = \ket{\underline{1}},
\end{equation}
leading to the non-trivial dynamics in $N$-dimensional Hilbert space of the $n^{(\uparrow)}$ sector in the $U(1)$-symmetric model both for Hermitian and non-Hermitian case. The initial state is then evolved via the \(\hat{\mathcal{P}}\hat{\mathcal{T}}\)-symmetric Hamiltonian \(\hat{\mathcal{H}}\) as \(|\Psi(t)\rangle = \mathcal{N}^{-1}e^{-i\hat{\mathcal{H}}t}|\Psi(0)\rangle\), where $\mathcal{N}=\sqrt{||\langle\Psi(t)|\Psi(t)\rangle||}$ is the normalization constant and $\hat{\mathcal{H}}\equiv \hat{\mathcal{H}}_{\text{nH}} (\hat{\mathcal{H}}_{\text{H}})$ for non-Hermitian (Hermitian) models. Note that the linear increase in Hilbert space with the system-size enables the investigation of systems with larger system-sizes.

The main goal is to create a highly entangled state between the links \(\ell_1\) and \(\ell_2\). Hence, we compute logarithmic negativity~\cite{peres_prl_1996, horodecki_pla_1996, vidal_pra_2002, plenio_prl}, denoted as \(\mathcal{E}_{\ell_1\ell_2} (t)\) obtained from the dynamical state \(|\Psi(t)\rangle\) after tracing out the bulk, i.e., $\rho_{1,N}(t)\!=\!\text{Tr}_B(\ket{\Psi(t)}\!\bra{\Psi(t)})$ (see Appendix ~\ref{app:log_neg}). To remove the time dependence from entanglement creation, we define the time-averaged link entanglement as 
\begin{equation}
\mathcal{E}_{\ell_1\ell_2}^{\text{avg}} = \frac{2}{T}\int_{T/2}^T \mathcal{E}_{\ell_1\ell_2}(t) dt
\end{equation}
with evolution for a large time $T$. Note that in our study, we choose \(T=5000\). We show that the Hamiltonian $\hat{\mathcal{H}}_{\text{nH}}$ with small link-interaction strengths $\lambda$ can generate maximally entangled links $(\ell_1 ,\ell_2)$ with $\mathcal{E}_{\ell_1\ell_2}^{\text{avg}}\approx 1$, when the bulk interacts via SSH model. In unitary evolution, the entanglement $\mathcal{E}_{\ell_1 \ell_2}(t)$ oscillates with time and near maximal entanglement is obtained, only at specific times, resulting in low values of average link entanglement $\mathcal{E}_{\ell_1\ell_2}^{\text{avg}}$. In contrast, due to non-unitary evolution via non-Hermitian Hamiltonian, the entanglement between the links is non-fluctuating and stabilizes at high value. 
In the following sections, we aim to study the LDE generation with the incorporation of non-Hermiticity in the model (open system dynamics) and compare it with the corresponding Hermitian (closed system dynamics) counterpart. Before presenting the results on the entanglement-creation, we first determine the exceptional points for this model which will be useful to prove one of the main results in the context of establishing maximal entanglement between the links \(\ell_1\) and \(\ell_2\). 

{\it Note.} While the entanglement can be distributed at large distances via quantum state transfer (QST) and has been studied when the interaction between the sites is of the $XX$~\cite{Bose_2003,  Wojcik05, Yao_2011, *Zwick_2011, *apollaro_pra_2012, *Vinet_2012, *Linneweber_2012, *Qin_2013, *Coden_2021} and the SSH types~\cite{Mei_2018, *DAngelis_2020, *Zhang_2023, *Hao_2025, *Xu_2024, *Ghosh_2025, ahuja_nhqst}, the dynamical generation of LDE from product state is a different architecture from QST scheme. In QST, pre-existing entanglement is transferred to the distant sites, whereas, here the generation of entanglement at long-distance, from a completely product state via time evolution is the main objective. A comparison of the two protocols is discussed in Appendix~\ref{app:xx_ent} for the $XX$ model.

\section{Exceptional points with small imaginary alternating magnetic fields and link-interactions}
\label{sec:ep_nH}

The \(\hat{\mathcal{P}}\hat{\mathcal{T}}\)-symmetric non-Hermitian SSH model $\hat{\mathcal{H}}_{\text{nH}}$ possess exceptional points where the system becomes defective. The exceptional points for $\hat{\mathcal{H}}_{\text{nH}}$ with weak link-interaction strengths $\lambda$, i.e., $\lambda \ll 1$ can be computed perturbatively for $\delta\!=\!1$. When the bulk possesses uniform interaction strength, i.e., the $XX$ model, the eigenvalues $E_a$ ($a \!=\! 1,2,\dots N)$ of the bulk $\hat{H}_{XX}$ in the $n^{(\uparrow)} \!=\! 1$ sector are given as
\begin{align}
    &\hat{H}_{XX} \!=\! \sum_{a=2}^{N-1} \!E_{a}\ket{\tilde{a}}\!\!\bra{\tilde{a}};\quad E_{a} \!=\! 4\cos\theta_a, \quad \theta_a \!\equiv\! \frac{(a\!-\!1)\pi}{N\!-\!1}, \nonumber\\
    &\ket{\tilde{a}} = \sqrt{\frac{2}{N\!-\!1}}\sum_{b=2}^{N-1} 
    \sin\left[(b\!-\!1)\theta_a\right] |\underline{b}\rangle.
    \label{eq:xx_spec}
\end{align}
for the bulk states ($a \!=\! 2,3\dots, N\!-\!1$), and $|\tilde{1}\rangle \!=\! |\underline{1}\rangle, |\tilde{N}\rangle \!=\! \ket{\underline{N}}$ with $E_1 \!=\! E_N \!=\! 0$ forms the zero energy eigenspace of the links $P_0 \!=\! |\tilde{1}\rangle\!\langle\tilde{1}| \!+\! |\tilde{N}\rangle\!\langle\tilde{N}|$. Note that $E_a \!=\! -E_{-a}$ with $-a\equiv N\!+\!1\!-\!a$ and we work with even $N$, which ensures no zero energy eigenstate in the bulk. Therefore, the off-diagonal terms are given as
\begin{align}
    &\hat{V}_{\ell} =  \sum_{a=2}^{N-1} 2v_a \Big[ (-1)^{a-1} (|\tilde{1}\rangle\langle\tilde{a}| + |\tilde{a}\rangle\langle\tilde{1}|) + |\tilde{N}\rangle\langle\tilde{a}| + |\tilde{a}\rangle\langle\tilde{N}| \Big], \nonumber\\
    &\hat{W} = 2\left(|\tilde{1}\rangle\langle\tilde{1}| - |\tilde{N}\rangle\langle\tilde{N}|+\sum_{a=2}^{N-1}\ket{-\tilde{a}}\bra{\tilde{a}}\right), \nonumber\\
    &\hat{\mathcal{H}}_{\text{nH}} (\delta\!=\!1) = \hat{H}_{XX} \!+\! \lambda\hat{V}_{\ell} \!+\! ih\hat{W},
    \label{eq:vw}
\end{align}
where $v_a\!=\!\sqrt{\frac{2}{N-1}}\sin\theta_a$, and the staggered magnetic field transforms an eigenstate \(|\tilde{a}\rangle\) with eigenenergy \(E_a\) into \(|-\tilde{a}\rangle\) with energy \(E_{-a} = - E_a\).

\begin{theorem}
    The exceptional points of the weakly interacting links with the $XX$ interacting bulk with imaginary staggered magnetic field, i.e., for the \(\hat{\mathcal{P}}\hat{\mathcal{T}}\)-symmetric $\hat{\mathcal{H}}_{\text{nH}}$, are given by $h_e(\delta\!=\!1, N, \lambda)=\lambda^2$ for $\lambda\ll 1$ and for even $N$.
    \label{thm:xx_ep}
\end{theorem}
\begin{proof}
    For small link-interactions, i.e., $\lambda\ll 1$, $\hat{\mathcal{H}}_{\text{nH}}(\delta\!=\!1)$ can be perturbatively expanded upto second order in $\lambda$ using Schrieffer-Wolff transformation~\cite{Schrieffer_1966, MacDonald_1988, Bravyi_2011, Starkov_2023}.  Here an appropriate generator $\hat{S}$ (proportional to the perturbation strength) is used such that the off-diagonal terms in the $P_0$ subspace vanishes upto the first order terms of the perturbation $\lambda$ (see Appendix~\ref{app:xx_sw} for derivation). The resulting interaction between the links is given by the transformed Hamiltonian in the $P_0$ subspace as
    \begin{equation}
        \hat{H}_{\ell_1 \ell_2}^{\text{SW}} = P_0 \left(e^{\hat{S}} \hat{\mathcal{H}}_{\text{nH}} e^{-\hat{S}}\right) P_0^{\dagger} =   
        2\left[\begin{matrix} 
            ih & (-1)^{\frac{N}{2}}\lambda^2 \\ 
            (-1)^{\frac{N}{2}}\lambda^2 & -ih 
        \end{matrix}\right].
        \label{eq:H_SW}
    \end{equation}
    with the eigenvalues $E^\prime_{a=1,N} =\pm2\sqrt{\lambda^4-h^2}$.
    The other eigenvalues from the bulk can be computed from the non-degenerate perturbation theory as
    \begin{align}
        E^\prime_{a} &= E_a+ \langle\tilde{a}|\lambda\hat{V}_{\ell} \!+\! ih\hat{W}|\tilde{a}\rangle \!+\! \sum_{\mathclap{k=2, k\neq a}}^{N-1} \frac{|\langle\tilde{a}|\lambda\hat{V}_{\ell} \!+\! ih\hat{W}|\tilde{k}\rangle|^2}{E_a - E_k} \nonumber\\
        &= 4\cos\theta_a + \frac{2\lambda^2 v_a^2}{\cos\theta_a} - \frac{h^2}{2\cos\theta_a},
        \label{eq:bulk_ee}
    \end{align}
    for $a \!=\! 2,3\dots, N\!-\!1$. Therefore, $E^\prime_{a\neq 1,N}$ are real, and the exceptional points are given by the eigenvalues of the link-interaction for $\lambda\ll 1$ which is obtained from first line of Eq. (\ref{eq:bulk_ee}). Hence, the system becomes defective at the staggered imaginary field,
    \begin{equation}
        h_e(\delta\!=\!1, N, \lambda) = \lambda^2,
        \label{eq:ep_nH}
    \end{equation}
    with the system being in the broken phase when $|h|>|h_e|$ for small $\lambda (\ll 1)$.
\end{proof}

On the other hand, for the non-Hermitian SSH model $\hat{\mathcal{H}}_{\text{nH}}(\delta\!\neq\!1)$, we find the exceptional points by numerically diagonalization of the system~\cite{arma2016}. Again, a similar relation with $h_e(\delta, N, \lambda) = a_2(\delta, N)\lambda^2 +a_1(\delta, N)\lambda$, provides the exceptional line for weak link-interaction strengths $\lambda\ll 1$. 
For $\delta\!<\!1$, both $a_2(\delta\!<\!1, N)$ and $a_1(\delta\!<\!1, N)$ increase as  the system-size $N$ (upto moderate $N$) grows, while $a_2(\delta\!<\!1, N)$ decreases with increasing $N$ for larger system-sizes and $a_1(\delta\!<\!1, N)$ saturates to a finite nonvanishing value. This results in quadratic exceptional line for small system-size $N$ and linear exceptional line for larger $N$ when $\delta<1$. For $\delta\!>\!1$, the exceptional points follow a quadratic scaling with $a_2(\delta\!>\!1, N)$ decreasing with increasing $N$, and $a_1(\delta\!>\!1, N)\sim10^{-4}$ (see details in Appendix~\ref{app:ssh_ep}).

\section{Role of non-Hermiticity in creating highly entangled links}
\label{sec:nH_dyn}

For weak bulk-link interactions ($\lambda \ll 1$), the oscillatory behavior of $\mathcal{E}_{\ell_1 \ell_2}(t)$ under unitary evolution (without the system-bath interactions) generates maximal LDE between the links at specific time intervals~\cite{Trifunovic_2013, Abaach2023, Soares_2025, agarwal2026_b}. This necessitates either terminating the protocol or decoupling the link from the bulk at an optimal time $t^*$, or waiting for a period $\sim 1/\lambda^2$, which becomes prohibitively large for weak interaction strengths \(\lambda\). Hence an immediate question arises, ``Is it possible to generate near-unit, and frozen long-distance entanglement over time with the non-Hermitian (open) systems?"

In this section, we aim to address this question and examine the role of imaginary alternating magnetic field on the LDE generation. The non-unitary dynamics is carried out through a \(\hat{\mathcal{P}}\hat{\mathcal{T}}\)-symmetric non-Hermitian NN model, given in Eq.~(\ref{eq:nhH_def}), for enhancement and stabilization of long-distance entanglement. We report the high value of entanglement, i.e., $\mathcal{E}_{\ell_1\ell_2}^{\text{avg}}\approx 1$, and determine the suitable tuning parameters leading to $\mathcal{E}_{\ell_1 \ell_2}^{\text{avg}}\approx 1$, thereby highlighting the impact of open quantum system dynamics in contrast to the closed unitary dynamics of the corresponding Hermitian model.

\begin{figure}
    \centering
\includegraphics[width=1.0\linewidth]{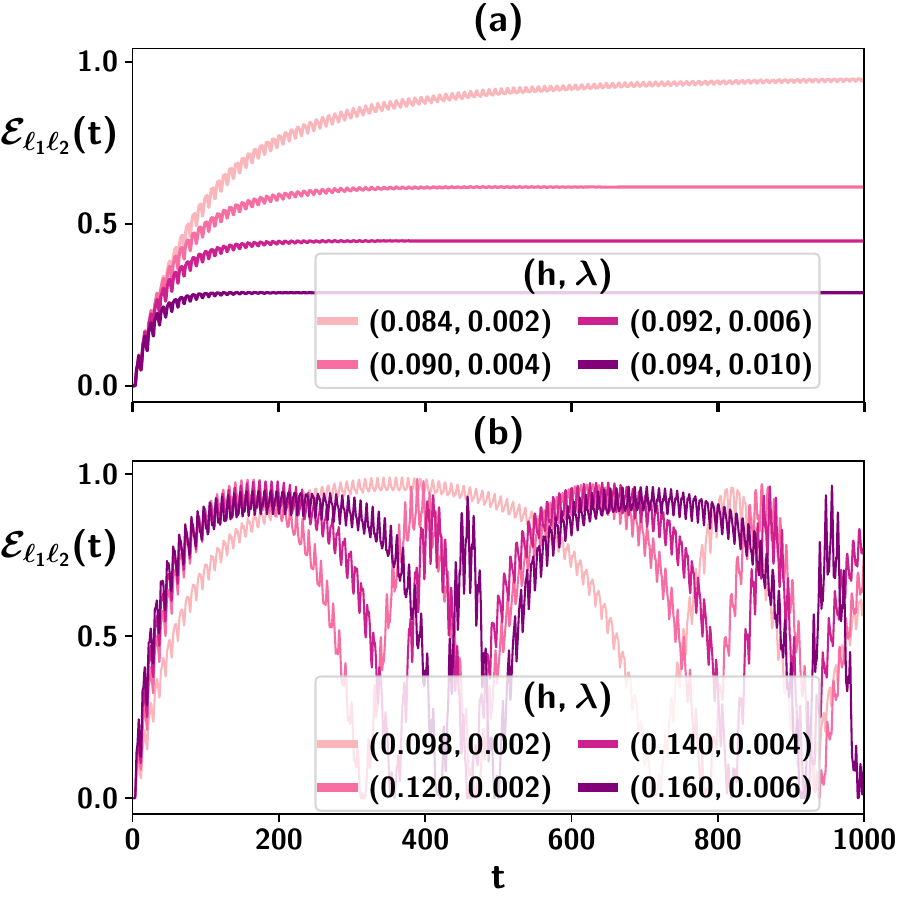}
    \caption{{\bf Comparison between long-distance entanglement \(\mathcal{E}_{\ell_1 \ell_2}\) of unbroken and broken regimes of the non-Hermitian SSH model.} \(\mathcal{E}_{\ell_1 \ell_2}\) (ordinate) vs time (abscissa) for \(\delta=1.2\) where (a) belongs to the broken regime and (b) to unbroken regime. The model facilitates long-distance entanglement creation in both regimes, with frozen entangled links in the broken regimes for a long period of time (a). However, the entanglement in the unbroken regime fluctuates, and the freezing periods are shorter. The system is initialized as \(|\Psi(0)\rangle = |1\rangle \otimes|0\rangle^{\otimes N-1}\) with the system-size \(N=16\). All the axes are dimensionless.}
    \label{fig:freezed_ent_vs_t}
\end{figure}

\begin{figure}
    \centering
\includegraphics[width=1.0\linewidth]{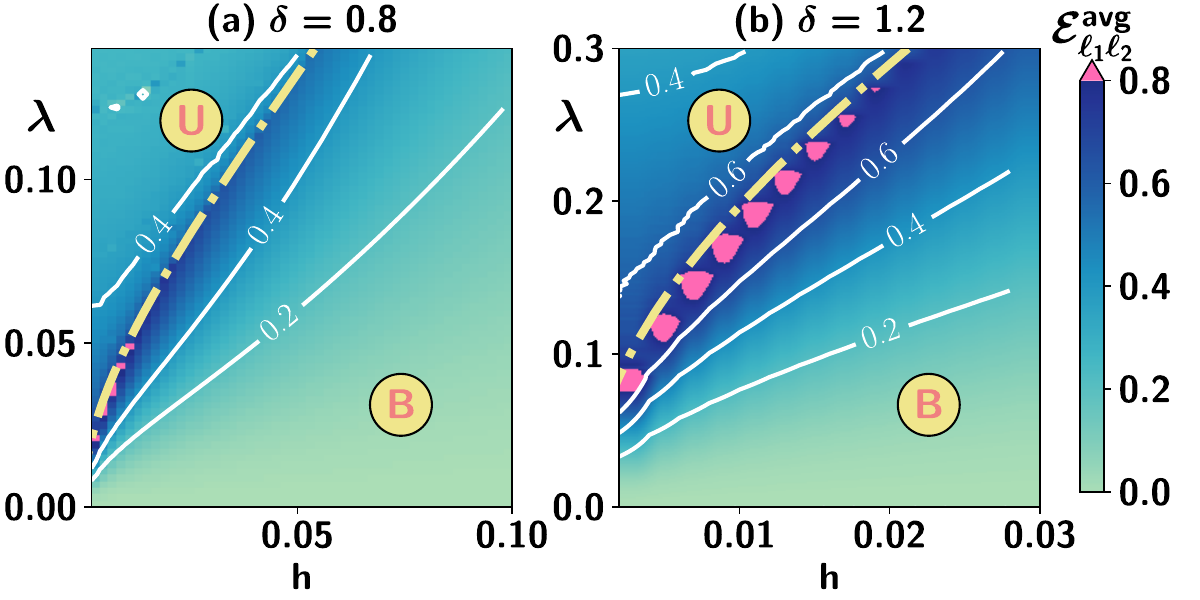}
    \caption{{\bf Time-averaged entanglement between the links \(\mathcal{E}_{\ell_1 \ell_2}^\text{avg}\): Exceptional points are special.} Time-averaged LDE \(\mathcal{E}_{\ell_1 \ell_2}^{\text{avg}}\) of the SSH model in the parameter space \((h,\lambda)\) (abscissa, ordinate) for (a) \(\delta=0.8\) and (b) \(\delta=1.2\), with \(N=16\). The exceptional points are represented by dashed-dot lines, which segregate the unbroken (U) and broken (B) phases. High LDE (\(> 0.8\)) (pink colored regions) can be obtained near the exceptional lines of the model, especially for \(\delta\geq 1\) (see Theorem 2, for \(\delta =1\)). The initial state of the bulk and link system is \(|\Psi(0)\rangle = |\underline{1}\rangle\). All the axes are dimensionless.}
    \label{fig:contour}
\end{figure}

\subsection{Entangled quantum links with \(\hat{\mathcal{P}}\hat{\mathcal{T}}\)-symmetric non-Hermitian NN SSH model}
\label{sec:dynamics_nn}

Let us explore the time dynamics of the LDE, \(\mathcal{E}_{\ell_1 \ell_2} (t)\) between the links by employing $U(1)$- and
\(\hat{\mathcal{P}}\hat{\mathcal{T}}\)-symmetric non-Hermitian SSH model as the evolving Hamiltonian, in the \(n^{(\uparrow)}\!=\!1\) sector. We consider the non-Hermitian SSH model in three regimes: (1) \( \delta<1 \), (2) \(\delta\!=\!1\) (the \(XX\) model) and (3) \(\delta>1\) for our analysis. Before investigating the regions (1) and (3), let us first exhibit that when \(\delta=1\), establishing a maximally entangled link is possible.

\begin{theorem}
    \label{thm:dyn}
    Maximally entangled links $(\ell_1, \ell_2)$, with $\mathcal{E}_{\ell_1\ell_2}^{\text{avg}} \!=\! 1$, are obtained by evolving the initial state $\ket{\Psi(0)} \!=\! \ket{\underline{1}}$ with the non-Hermitian \(\hat{\mathcal{P}}\hat{\mathcal{T}}\)-symmetric $XX$ model having weak link-interactions $\lambda\ll1$ near the exceptional points in the broken phase at long time limit. 
\end{theorem}
\begin{proof}
    For an initial state $\ket{\Psi(0)} \!=\! \sum_{a=1}^{N} c_a |\tilde{a}^\prime\rangle$ in eigenbasis of  $\hat{\mathcal{H}}_{\text{nH}}(\delta\!=\!1)$ in the $n^{(\uparrow)} \!=\! 1$ sector, the dynamical state $\ket{\Psi(t)}$ can be expressed as
    \begin{equation}
        \ket{\Psi(t)} \!=\! \frac{1}{\mathcal{N}(t)}\sum_{a=1}^{N} c_a e^{-iE_a^{\prime}t} |\tilde{a}^\prime\rangle, \nonumber
    \end{equation}
    where $\mathcal{N}(t)$ is the normalization factor. The eigenvalues of $\hat{\mathcal{H}}_{\text{nH}}(\delta\!=\!1)$ for weak link-interaction strengths $\lambda\ll1$ are given by $E^{\prime}_{a=1,N} \!=\! \pm2\sqrt{\lambda^4\!-\!h^2}$ for states in the $P_0$ subspace and by Eq.~(\ref{eq:bulk_ee}) for the bulk states. Denoting $E^{\prime}_{1} \!\equiv\! E^{\prime}_{-}$, $E^{\prime}_{N} \!\equiv\! E^{\prime}_{+}$, the corresponding eigenvectors are $\ket{\phi_{\pm}} \!=\! \frac{1}{\mathcal{N}_{\pm}}\left[c^{(\pm)}_1, 0,\cdots ,0,c^{(\pm)}_N\right]^T$, where $c^{(\pm)}_1 \!=\!ih \!\pm\! E^{\prime}_{\pm}$, $c^{(\pm)}_N \!=\!\lambda^2$ and $\mathcal{N}_{\pm} \!=\!\sqrt{2}\lambda^2$ is the normalization factor, when $h< \lambda^2$, i.e., in the unbroken phase. On the other hand, when the system is in the broken phase with $h> \lambda^2$, $\mathcal{N}_{\pm} \!=\!\sqrt{2h(h\pm|E^{\prime}_{\pm}|)}$. At the exceptional point, $h \!=\! h_{e} \!=\! \lambda^2$, the eigenvectors coalesce with $\ket{\phi^\prime_{+}} \!=\! \ket{\phi^\prime_{-}}$. Note that in the unbroken phase and at the exceptional point, $|c^{(\pm)}_1|^2 \!=\! |c^{(\pm)}_N|^2$, while $|c^{(\pm)}_1|^2 \!\approx\! |c^{(\pm)}_N|^2$ in the broken phase and near the exceptional points, giving near-maximally entangled links as the eigenstates of \(\hat{\mathcal{H}}_{\text{nH}}\) (see Appendix~\ref{app:log_neg}).
    
    Therefore, during evolution within the broken phase where $h>\lambda^2$, the eigenvalues $E^\prime_{\pm}$ become imaginary, causing the system to approach $\ket{\phi^\prime_{-}}$ exponentially over time. Consequently, for non-zero coefficients $c_a, c_N \!\neq\! 0$, the state $\ket{\Psi(t \!\to\! \infty)} \!=\! \ket{\phi^\prime_{-}}$ emerges as the stationary state. Because $\ket{\phi^\prime_{-}}$ is maximally entangled near the exceptional points, highly entangled links are established in the long-time limit of the broken phase. Specifically, starting from the initial state $\ket{\Psi(0)} \!=\! \ket{\underline{1}}$ (where $c_1\approx1$ for $\lambda\ll1$), the entanglement, quantified for example by logarithmic negativity~\cite{vidal_pra_2002} approaches unity ($\mathcal{E}_{\ell_1\ell_2}\to1$) as the system nears the transition ($h-\lambda^2\to0^{+}$), yielding stable, high entanglement.
\end{proof}
    
In the unbroken phase with $h<\lambda^2$, all the eigenvalues are real, and the evolution is oscillatory with time, resulting in fluctuating $\mathcal{E}_{\ell_1 \ell_2}(t)$.  In the unbroken regime, while $\mathcal{E}_{\ell_1\ell_2}(t)\sim1$ is obtained at specific times, $\mathcal{E}_{\ell_1\ell_2}^{\text{avg}}\!<\!1$ due to the fluctuating behavior.

Let us now move to study the entanglement generation in the dynamics of the $U(1)$- and
\(\hat{\mathcal{P}}\hat{\mathcal{T}}\)-symmetric non-Hermitian SSH model with $\delta\!\neq\!1$, i.e., for $\delta\!>\!1$ and $\delta\!<\!1$. Such a study is possible only through extensive numerical simulations as discussed before upto system-size \(N\leq100\). In-depth analysis in this regime reveals some interesting features of non-Hermiticity in the LDE profile and, in addition, shows benefits over the study in the corresponding Hermitian counterpart. Note that (\(\lambda=0\)) leads to unentangled links by definition. Further, we find that the system undergoes a transition from the broken to the unbroken regime in the \((h,\lambda)\)-plane. As discussed in Sec. \ref{sec:framework}, when \(\lambda\) is small, which is typically assumed in this architecture, the exceptional points can be found for a very weak magnetic field, i.e., the distinctive behavior of entanglement in the unbroken and broken regimes can be scrutinized only in the presence of a weak magnetic field.

{\it Observation 1: Frozen dynamical entanglement in the broken regime.} When the system is in the broken region, link entanglement fluctuates, but the fluctuation is very small, and it increases with time before converging to a fixed value. Specifically, this converged value near \(\mathcal{E}^{\text{avg}}_{\ell_1 \ell_2}\) depends on the choice of \((h,\lambda)\)-pair. The high value of entanglement can be obtained when the system is close to the exceptional points (see Fig. \ref{fig:freezed_ent_vs_t}(a) and Table \ref{table:1}). In contrast, in the unbroken regime, entanglement fluctuates as well as oscillates with time and does {\it not converge} to a fixed value. To capture the fluctuating behavior of the link entanglement, we compute the average standard deviation, \(\int_{T/2}^T \sqrt{\langle\mathcal{E}_{\ell_1 \ell_2}^2\rangle - \langle\mathcal{E}_{\ell_1 \ell_2}\rangle^2} dt\). We find that in the unbroken region, it is \(\gtrsim\mathcal{O}(10^{-1})\)  while for the broken one, it is almost negligible, i.e.,  \(\lesssim\mathcal{O}( 10^{-3})\), thereby capturing the contrasting trend of unbroken-broken regimes (see Fig. \ref{fig:freezed_ent_vs_t}).

However, in both cases, there are certain time regions at which \(\mathcal{E}_{\ell_1 \ell_2}(t)\) can achieve almost unit entanglement, although average entanglement in the unbroken regime never crosses a certain threshold value. Further, the high converged value can be found only when the system is close to the exceptional line. More precisely, when the system approaches the exceptional line from the broken phase, the \(\mathcal{E}_{\ell_1 \ell_2}^{\text{avg}} \approx 1\) while similar neighborhood in the unbroken phase can produce \(\mathcal{E}_{\ell_1 \ell_2}^{\text{avg}} < 1\).

\begin{table}[h!]
\begin{center}
\begin{tabular}{||c@{\hspace{1.2cm}}c@{\hspace{1.2cm}}c@{\hspace{1.2cm}}c@{\hspace{0.2cm}}||}
 \hline
 \quad\(\delta\) & \(\lambda\) & \(h\) & \(\mathcal{E}_{\ell_1 \ell_2}^{\text{avg}}\ \) \\ [0.5ex] 
 \hline\hline

\quad 0.7 & 0.012 & 0.002 & 0.827\quad  \\
\hline
\quad 0.8 & 0.02 & 0.002 & 0.917 \quad  \\
\hline
\quad 1.0 & 0.044 & 0.002 & 0.948\quad  \\
\hline
\quad 1.0 & 0.06 & 0.004 & 0.875\quad   \\
\hline
\quad 1.2 & 0.084 & 0.002 & 0.957\quad  \\
\hline
\quad 1.4 & 0.146 & 0.002 & 0.961\quad  \\
 \hline
\end{tabular}
\caption{The dependency of time-averaged value of the link entanglement, i.e., \(\mathcal{E}_{\ell_1 \ell_2}^{\text{avg}}\) on system parameters \((h,\lambda)\), near the exceptional points. Here, the system size is \(N=16\).}
\label{table:1}
\end{center}
\end{table}

{\it Observation 2: Neighborhoods of exceptional points are special.} As discussed in Observation 1, high time-averaged link entanglement can only be obtained when the system is close to the exceptional point for all values of \(\delta\) and in the broken phase. However, in the broken phase, due to low fluctuations, the average entanglement produced is relatively higher than the unbroken phase. For instance, we find that when the absolute sum of the imaginary eigenvalues of the Hamiltonian, ensuring that the system is in broken phase, is less than \(\mathcal{O}(10^{-3})\), very high average entanglement between links is created. Almost maximally entangled link state can be produced when the system is in the limit of exceptional point \footnote{We note here that the exceptional points are system-size independent when \(\delta =1\), i.e., for the \(XX\) model while it depends on the system-size for \(\delta\neq1\). However, the features reported here remain same, independent of system-size.}.

{\it Observation 3: Non-Hermitian vs Hermitian.} 
To probe the benefit of open quantum system, leading to non-Hermitian systems over Hermitian counterparts, we introduce a quantity called maximum time-averaged link entanglement defined as
\begin{equation}
\mathcal{E}_{\max}^{\text{avg,X}} = \max_{h,\lambda} \mathcal{E}_{\ell_1 \ell_2}^{\text{avg}}, \quad \text{X=nH,H}
\label{eq:max_timeavg}
\end{equation}
where averaging is performed over time, and maximization is carried out in the \((h,\lambda)\)-plane. For a fixed and moderate bulk-size, we observe that when \(\delta \approx 1\) and when \((h,\lambda)\) values are close to the exceptional points, the maximum time-averaged entanglement between \(\ell_1\) and \(\ell_2\) achieves a near maximal value. Precisely, there exists a critical bulk-size above which entanglement starts decreasing and these critical values depend on \(\delta,h\), and \(\lambda\) values. In contrast, when the evolution is governed by the Hermitian model, the overall trend of the link entanglement follows a similar behavior as in the non-Hermitian case, although it never attains the maximum value. In particular, for a fixed \(N\), we find \(\mathcal{E}_{\max}^{\text{avg,nH}} \geq \mathcal{E}_{\max}^{\text{avg,H}}\), as depicted in Fig. \ref{fig:scal_dyn}. For example, with \(\delta=1.2\), \(N=16\), \(\mathcal{E}_{\max}^{\text{avg, H}} =0.765\) (with \(\lambda=0.07,h=0.001\)), while \(\mathcal{E}_{\max}^{\text{avg, nH}} =0.987\) with \(\lambda=0.06,h=0.001\)). It indicates that the effective non-Hermitian model obtained due to the system reservoir interaction can be more advantageous to create entanglement between distant links through bulk than the corresponding Hermitian isolated model.

{\it Observation 4: \(XX\) model is best among the class of SSH models.} If one considers maximum time-averaged link entanglement defined in Eq. (\ref{eq:max_timeavg}), we find that near exceptional point, \(\mathcal{E}_{\max}^{\text{avg,nH}}\) with \(\delta=1\) provides maximal entanglement compared to the evolving Hamiltonian with \(\delta \neq1\) (see Fig. \ref{fig:scal_dyn}). There exists \(\delta (>1)\) for which \(\mathcal{E}_{\max}^{\text{avg,nH}}(\delta>1) < \mathcal{E}_{\max}^{\text{avg,nH}}(\delta<1)\) for a fixed system-size (eg. comparing \(\mathcal{E}_{\max}^{\text{avg,nH}}(\delta=1.4) =0.936\) and \(\mathcal{E}_{\max}^{\text{avg,nH}}(\delta=0.5)=0.321\) for \(N=20\)) while the opposite ordering can also be found for pairs of \(\delta>1\) and \(\delta<1\). The entire analysis clearly indicates that the \(XX\) model with \(\delta=1\) with suitable pairs \((h,\lambda)\) can produce maximal link entanglement, which the other classes of SSH models fail to generate.

{\it Observation 5: Scaling with system-size: } The time-averaged link entanglement, maximized over \((h,\lambda)\)-plane, decays with the system-size \(N\) for all \(\delta\) values. Upto moderate system-size, i.e., \(N \leq 50\), the \(\mathcal{E}_{\max}^{\text{avg},X}\) decays slowly in the neighborhood of \(\delta=1\). Specifically, for the non-Hermitian model, the scaling for \(\delta=1\) is \(N^{-0.034}\), whereas for \(\delta\) values in the neighborhood of unity, i.e., \(\delta=0.9\), it is \(N^{-0.06}\) and for \(\delta=1.1\), it is \(N^{-0.026}\) (see Fig. \ref{fig:scal_dyn} (a)). Moreover, when \(N>50\), the decrease is slow in the case of the \(XX\) model as compared to the cases of \(\delta \neq 1\), thereby establishing the advantage of the \(XX\) model in link entanglement creation in open systems. Note that the maximum time-averaged entangled link in the corresponding Hermitian model decreases with increasing system-size as well, in a similar fashion to the non-Hermitian case (see Fig. \ref{fig:scal_dyn} (b)), although the initial values of the Hermitian case are much lower than those of the non-Hermitian ones.

{\it Observation 6: Initial state dependence.} The results presented above are for a particular initial state \(\ket{\Psi(0)} = |\underline{1}\rangle\) for both the non-Hermitian and Hermitian SSH models. It is now natural to ask: ``Does the entanglement generation depend on the initial states?" The answer turns out to be more involved. If we restrict to the initial state lying in the $\{n^{(\uparrow)}\!=\!0,1\}$ sectors with non-vanishing overlap with the eigenstates of the links, i.e., \(|\Psi(0)\rangle = \ket{\psi^{\theta_1, \phi_1}}_{\ell_1}\otimes \ket{\downarrow}^{\otimes N-2}\otimes\ket{\psi^{\theta_2, \phi_2}}_{\ell_2}\)  \( \left(\ket{\psi^{\theta_i, \phi_i}} \!=\! \cos\frac{\theta_i}{2}\ket{\downarrow} \!+\! e^{i\phi_i}\sin\frac{\theta_i}{2}\ket{\uparrow}\right)\) (i=1,2) with \(\theta_{1(2)}\!\!=\!0\) and arbitrary \(\theta_{2(1)}\), the link entanglement generation discussed above trivially remains same (see Theorem~\ref{thm:dyn} for the $XX$ bulk). If we analyze two initial states in the $\{n^{(\uparrow)}\!=\! 0, 1, 2\}$ sectors, the dimension of the Hilbert space become \(1\), \(N\) and \(N(N\!-\!1)/2\) respectively. Specifically, let us consider the initial states \(\ket{\Psi(0)}_{11} =\ket{\uparrow}_{\ell_1} \otimes \ket{\downarrow}_{B}^{\otimes N-2}\ket{\uparrow}_{\ell_2}\) and \(\ket{\Psi(0)}_{++} =\ket{+}_{\ell_1} \otimes \ket{\downarrow}_{B}^{\otimes N-2}\ket{+}_{\ell_2}\) where \(\ket{+} = \frac{1}{\sqrt{2}} (\ket{\uparrow} + \ket{\downarrow})\). We find the time-averaged entanglement in the broken regime, maximized over \((h,\lambda)\)-plane, is slightly lower when $n^{(\uparrow)}\!=\!2$ sector, as compared to the dynamics in $n^{(\uparrow)}\!=\!1$ only. For example, with system-size \(N\!=\!16\) and the $XX$ model, \(\mathcal{E}_{\max}^{\text{avg,nH}} = 0.954\) for \(\ket{\Psi(0)}_{11}\) and \(\mathcal{E}_{\max}^{\text{avg,nH}} = 0.93\) for \(\ket{\Psi(0)}_{++}\), whereas for \(\ket{\Psi(0)}=|\underline{1}\rangle\), \(\mathcal{E}_{\max}^{\text{avg,nH}} = 0.958\) is slightly larger (of $O(10^{-3})$ order). Therefore, the results presented for the $n^{(\uparrow)}\!=\!1$ sector are valid for $n^{(\uparrow)}\!=\!2$ as well, with near maximally entangled links obtained near the exceptional points. Hence, the behavior of link entanglement generation with continuously monitored baths is robust and not specific to the initial state of the links \((\ell_1,\ell_2)\).

\subsection{The role of long-range interactions on LDE with non-Hermitian SSH model}
\label{sec:dynamics_lr}

After demonstrating the benefit of non-Hermiticity in the link-to-link entanglement creation over the Hermitian case, an immediate question arises: ``Does the benefit of non-Hermiticity in the entanglement generation persist if the long-range interactions are present in the evolving Hamiltonian?" We answer this question affirmatively. To establish this, we consider the long-range \(XX\) Hamiltonian in the bulk and in the link as
\begin{align}
 \nonumber \hat{H}_{B} &= \sum_{a=2}^{N-2}\sum_{b=a+1}^{N-1} \frac{1}{|a-b|^{\alpha}}\hat{\mathsf{H}}_{a, b},
 \\  
\nonumber \hat{V}_{\ell} &= \lambda\sum_{a=1,N}\sum_{b=2(\neq a)}^{N} \frac{1}{|a-b|^{\alpha}}\hat{\mathsf{H}}_{a, b}
\end{align} 
where \(\alpha\) is the power law decay strength in the spin-spin couplings. Note that the interactions of the links (\(\ell_1,\ell_2\)) with the bulk spins get diminished by the factor of \(\frac{\lambda}{|a-b|^{\alpha}}\), i.e., based on the separation between the spins. Such power law decaying variable-range interacting models naturally appear in several physical platforms \cite{long_range_review,bloch05,blochrmp} and hence such consideration in this context ensures that the scheme presented here is also independent of the physical architecture. The fall-off rate \(\alpha \geq 5\) mimics the nearest-neighbor \(XX\) model while \(\alpha <3 \) typically indicates the presence of long-range interactions in the system. In particular, there are three regimes of \(\alpha\) \cite{philip_taglia,GoroshkovLR}, in the thermodynamic limit for the variable-range interacting Ising model \cite{entropy_long_range_maciej_luca_prl_2012,VodolaLR}: (1) long-range  (\(0 \leq \alpha \leq 1\)), (2) quasi long-range (\(1 \leq \alpha \leq 2\)) and (3) short-range (\(\alpha \geq 2\)). Note, however, that these standard regimes are valid in the thermodynamic limit of system-size and since in our case, moderate system sizes are important from the implementation point of view, one can argue that these regimes are subject to change.

\begin{figure}
    \centering
\includegraphics[width=1.0\linewidth]{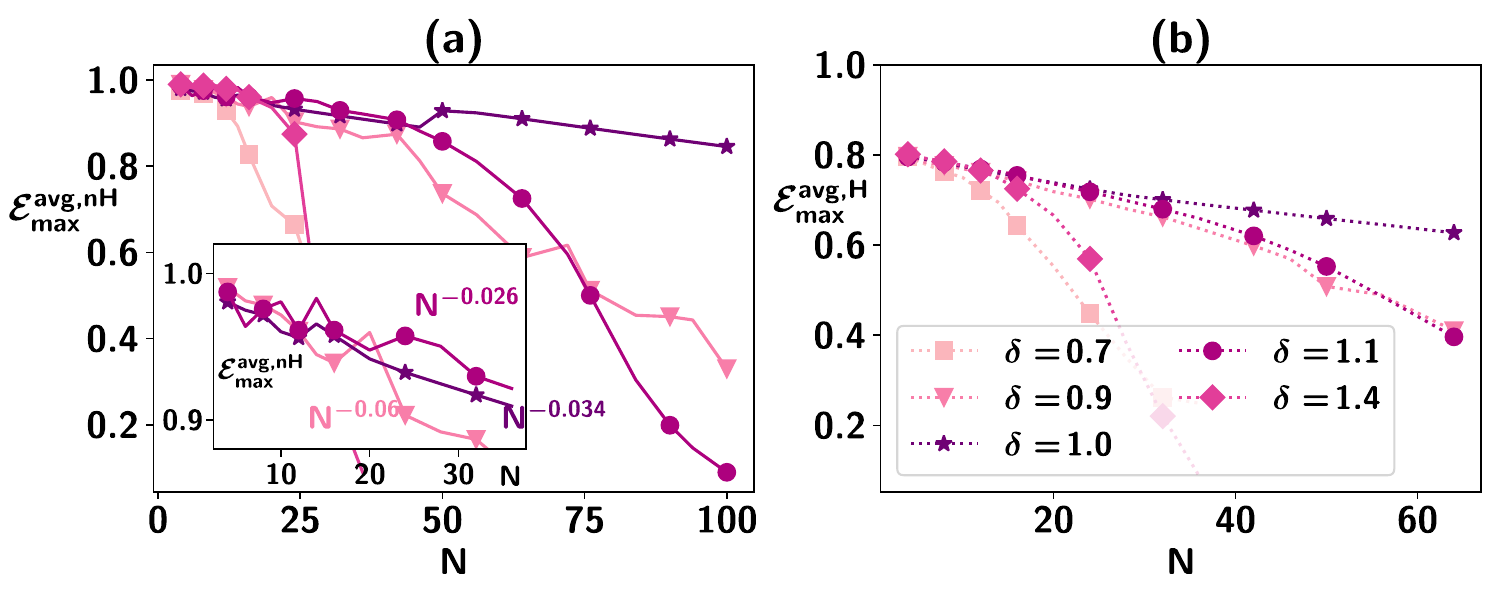}
    \caption{ {\bf Scaling of LDE with system-size: Comparison between non-Hermitian and Hermitian LDE generation. } The entanglement \(\mathcal{E}_{\max}^{\text{avg},X}\) (ordinate), averaged over time and maximized over the \((h,\lambda)\)-plane with respect to the system-size \(N\) (abscissa) for (a) non-Hermitian (X=nH) and (b) the Hermitian (X=H) models. Different lines denote different \(\delta\) values of the SSH model.  \(\mathcal{E}_{\max}^{\text{avg},X}\) decays slowly when \(\delta \approx 1\), whereas the LDE decays quickly with system-size $N$ when \(\delta \neq 1\) and high $N$, highlighting the significance of \(XX\) model. (Inset) When entanglement is close to unity (\(\mathcal{E}_{\max}^{\text{avg},nH} <0.9\)), the scaling exponents are found for \(\delta =0.9, 1.0 \) and \(1.1\). All the axes are dimensionless. }
    \label{fig:scal_dyn}
\end{figure}


\begin{figure}
    \centering
\includegraphics[width=0.9\linewidth]{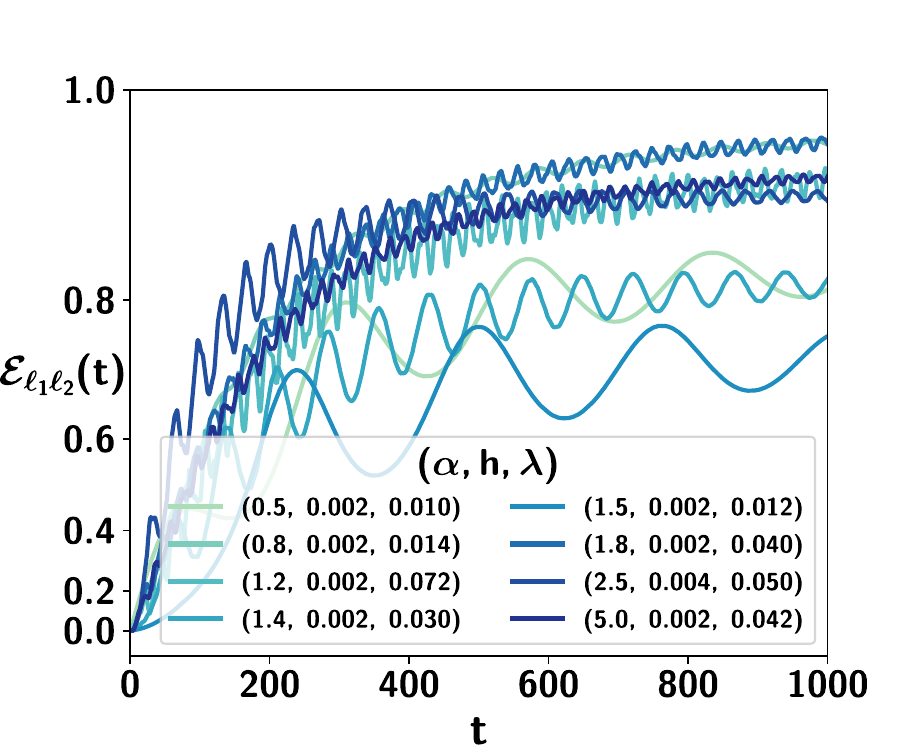}
    \caption{ {\bf Dynamical entanglement pattern with long-range interactions. }  \(\mathcal{E}_{\ell_1 \ell_2}\) (ordinate) vs time \(t\) (abscissa)  for different values of (\(\alpha, h, \lambda\))  in the \(XX\) model. For a fixed value of \(\alpha\), the pairs of \((h,\lambda)\)  are chosen in such a way that the time-averaged entanglement \(\mathcal{E}_{\ell_1 \ell_2}^{\text{avg}}\) is maximum. The maximization is performed when \(h \in [0,0.1]\) and \(\lambda \in [0,0.3]\) for \(N=16\). Specifically, LR interactions with \(\alpha =0.8\) and \(\alpha =1.8\) clearly provide benefits over the nearest-neighbor case (\(\alpha =5.0\)).  All the axes are dimensionless.}
    \label{fig:dyn_LR}
\end{figure}

{\it Observation 7: Advantage in long-range perturbations.} Following the same protocol as described for the NN model, we study \(\mathcal{E}_{\ell_1 \ell_2} (t)\) with the variation of time for different \(\alpha\) values. First of all, the overall nature of entanglement, i.e., monotonic increase of entanglement with fluctuations at initial times and saturating at large times remains the same, irrespective of \(\alpha\) values. The fluctuations as well as saturation value of entanglement depend on low values of \(\alpha \in [0,2]\). Secondly, we also observe that, like the NN model, the time-averaged entanglement, maximized over \((h,\lambda)\)-plane, occurs close to the exceptional line when the long-range interactions are present in the system. Further, subject to the choice of \((h,\lambda)\), there exist some \(\alpha\) values in  long-range and quasi long-range regimes, where the saturated value of \(\mathcal{E}_{\ell_1 \ell_2}(t)\) is higher than that can be obtained through the evolving Hamiltonian with nearest-neighbor interactions (see Fig. \ref{fig:dyn_LR}). Therefore, the long-range interactions clearly can provide benefit in the link entanglement creation over the nearest-neighbor ones.

In addition, for a particular choice of \(\alpha\), when the time-averaged entanglement is maximized over \((h,\lambda)\)-plane, non-Hermitian model again turns out to be more effective for high entangled link production than the corresponding Hermitian case, for moderate system sizes.

\begin{figure}
    \centering
\includegraphics[width=1.0\linewidth]{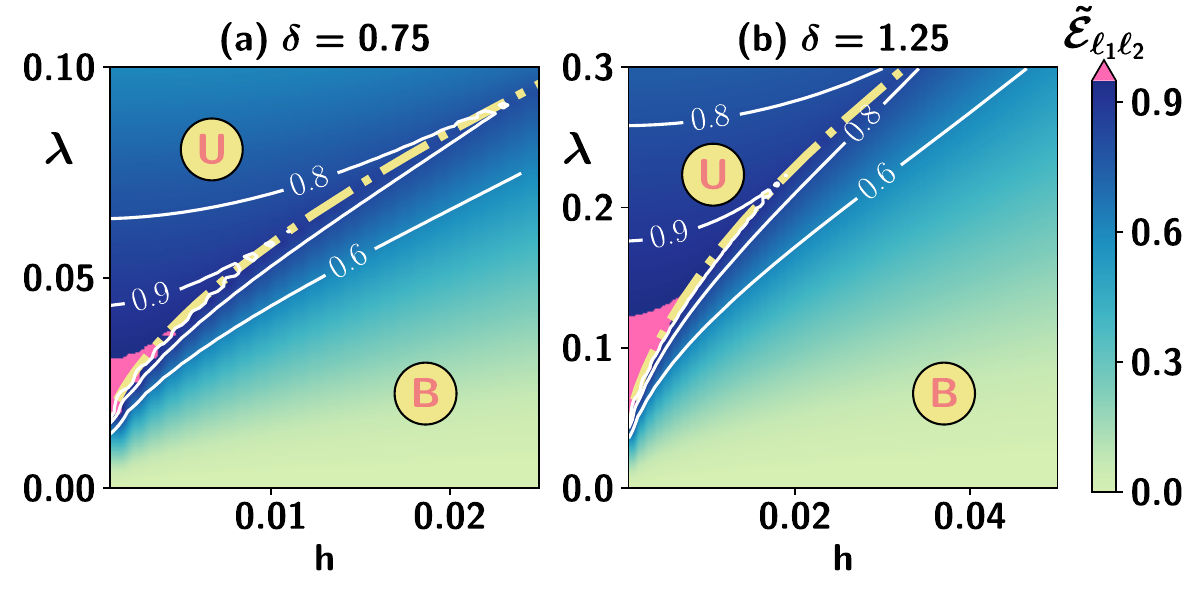}
    \caption{ { \bf Long-distance entangled links in a stationary state.} The variation of entanglement between the links $\tilde{\mathcal{E}}_{\ell_1 \ell_2}$ in the \((h,\lambda)\)-plane is presented when the system reaches the stationary state for (a) \(\delta=0.75\) and (b) \(\delta=1.25\), with system-size \(N=10\). The broken (B) and unbroken (U) regimes are separated by the dashed-dot lines. High entanglement between the links \((\ell_1,\ell_2)\) is produced close to the exceptional line with pink colored regions indicating \(\tilde{\mathcal{E}}_{\ell_1 \ell_2} \geq 0.95\).  All the axes are dimensionless.}
    \label{fig:stat_xx_nnlr}
\end{figure}
\section{Entangled quantum links in stationary states}
\label{sec:statics}

The dynamical generation of highly entangled links via \(\hat{\mathcal{P}}\hat{\mathcal{T}}\)-symmetric non-Hermitian Hamiltonian raises a question on the behavior of entangled links in the limit of a stationary eigenspectrum:  Can the system sustain high LDE at the stationary situation with the presence of non-Hermiticity? The affirmative answer then indicates that even before starting the dynamics, the initial state between the links can be highly entangled. This question can be addressed by examining the stationary states, which are the eigenstates obtained when the system is connected to a thermal reservoir with vanishing temperature. 
The ground states of the Hermitian models have been extensively studied in the limit of a weak interaction between the bulk and the link. It is shown that the links separated by the bulk can be highly entangled at zero temperature~\cite{,venuti_prl, VenutiPRA07, Ferreira_PRA_2008, Giampaolo_2010, Dhar_2016, agarwal2026_b}. In the non-Hermitian Hamiltonian, this architecture can give rise to control of two distinct parameter regions, unbroken and broken.\\
(1) {\it Unbroken regime: } In the unbroken regime, since all the eigenvalues of the Hamiltonian are real, the stationary state corresponds to the least energy eigenstate (ground state) of the Hamiltonian. \\
(2) {\it Broken regime: } As the eigenvalues are in the complex plane, we define the stationary state as the eigenstate associated with the eigenvalue that has the largest imaginary part. If multiple eigenvalues share this largest imaginary component, we further select the one with the smallest real part. Thus, the time \(t \rightarrow \infty\) state corresponds to an eigenvalue of the form ($a + ib$), where $b$ is maximized and, among those, the value $a$ is minimized. 

As discussed in Theorem~\ref{thm:dyn}, this definition of the stationary state in the broken regime corresponds to the long-time limit \(t \rightarrow \infty\), due to the exponential growth of the coefficient associated with the eigenstate having the largest imaginary part. Therefore, in the broken regime, the LDE of the stationary state corresponds to the dynamically generated LDE discussed in the previous section. In contrast, for the unbroken regime, the LDE of the lowest-energy eigenstate serves as an alternate route to obtain high and nonfluctuating entangled links. Note that while the system is still $U(1)$-symmetric, the stationary state can belong to any of the $n^{(\uparrow)}$ sectors, and in the unbroken regime, the stationary state is obtained from the $n^{(\uparrow)}=N/2$ sector. Therefore, the numerical analysis is restricted by the exponential growth of Hilbert space, and we illustrate the results with system-sizes $N\leq 10$.


The stationary state of the \(\hat{\mathcal{P}}\hat{\mathcal{T}}\)-symmetric non-Hermitian SSH Hamiltonian $\hat{\mathcal{H}}_{\text{nH}}$ in Eq.~(\ref{eq:nhH_def}) provides highly entangled links, which is quantified by $\tilde{\mathcal{E}}_{\ell_1 \ell_2}$, to distinguish from the dynamical framework. Specifically in the unbroken regime with $\lambda\ll1$ (and close to the exceptional line), the LDE in the stationary state is higher than in the broken regime, which is in contrast with the dynamical scenario. For example, for system-size \(N\!=\!10\) and $\delta\!=\!1.0$, \(\tilde{\mathcal{E}}_{\ell_1\ell_2} = 0.993\), when \((h,\lambda)\!=\!(0.001,032)\), which occurs in the unbroken phase close to the exceptional line $h_e=\lambda^2$ of the $XX$ model. Similar behavior of maximum entanglement between the links is seen for the SSH model, as shown in Fig.~\ref{fig:stat_xx_nnlr} for $\delta=0.75$ and $1.25$.

\begin{figure}
    \centering
\includegraphics[width=0.8\linewidth]{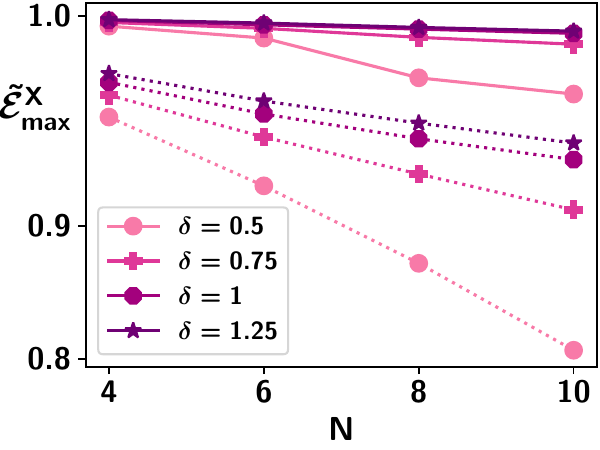}
    \caption{ {\bf System-size scaling of maximum LDE in the stationary state. } The maximum LDE \(\tilde{\mathcal{E}}_{\max}^{X}\) (ordinate) against system-size \(N\) (abscissa) for various \(\delta\) values, where \(\text{X}=\text{nH} (\text{H})\) represents non-Hermitian (Hermitian) case with solid (dotted) lines. The non-Hermitian SSH model provides benefit as $\tilde{\mathcal{E}}^{\text{nH}}_{\max}>\tilde{\mathcal{E}}^{\text{H}}_{\max}$ and $\tilde{\mathcal{E}}^{\text{nH}}_{\max}$ decreases slowly with  \(N\) as compared to $\tilde{\mathcal{E}}^{\text{H}}_{\max}$ for all \(\delta\) values considered. All the axes are dimensionless.}
    \label{fig:scaling_st}
\end{figure}

{\it Observation 8: Non-Hermitian vs Hermitian in stationary states.} To investigate the role of the non-Hermiticity in the stationary framework, we again quantify the maximum entanglement between the links by 
\begin{equation}
    \tilde{\mathcal{E}}_{\max}^{\text{X}} = \max_{h,\lambda} \tilde{\mathcal{E}}_{\ell_1 \ell_2}, \quad \text{X=nH,H}
    \label{eq:max_stationary}
\end{equation}
where maximization is over \((0\!<\!h\!\leq\! 0.1, 0\!<\!\lambda\!\leq\!0.3)\) parameter space. For a given $\delta$ and fixed system-size $N$, the maximum entanglement between the links is obtained for a given value of $(h, \lambda)$ parameters, which decreases with the system-size.
While the ground state of the Hermitian model contain high LDE, the maximum LDE in the stationary states of non-Hermitian model is higher, for example, for $N\!=\!10$ and $\delta\!=\!1.0$, \(\tilde{\mathcal{E}}_{\max}^{H}=0.937<\tilde{\mathcal{E}}_{\max}^{nH}=0.993\). 
Such a benefit of non-Hermiticity is seen for various values of $\delta$, as shown in Fig.~\ref{fig:scaling_st}. It clearly depicts
\begin{equation}
    \tilde{\mathcal{E}}_{\max}^{nH}>\tilde{\mathcal{E}}_{\max}^{H},
\end{equation}
in the stationary framework (comparing solid and dashed lines in Fig. ~\ref{fig:scaling_st}).
Interestingly, while the maximum LDE $\tilde{\mathcal{E}}_{\max}^{\text{X}}$ decreases with increasing system-size $N$, the decrease is much smaller for the non-Hermitian case than the Hermitian one (see Fig.~\ref{fig:scaling_st}).

Further, the rate of decay in the stationary (ground) state of the non-Hermitian (Hermitian) model with system-size is smaller with the increase of $\delta$ $(N\leq 10)$, thereby indicating higher robustness of LDE with higher $\delta$.

Our analysis, both in stationary and dynamical situations, illustrates that the links interacting with the bulk and bath, along with weak continuous measurements on the auxiliary qubits, lead to high link entanglement creation as compared to its Hermitian counterpart. 





\section{Conclusion}
\label{sec:conclu}

Highly entangled states shared over large distances serve as essential resources for scalable quantum networks and various quantum information processing tasks. In this context,  interacting spin systems can serve as quantum buses, through which entangling quantum links weakly coupled to their boundaries are created, thereby providing a mechanism for creating long-distance entanglement. In realistic scenarios, quantum systems inevitably interact with their environments; however, continuous monitoring of this bath can effectively induce non-unitary dynamics governed by an effective non-Hermitian Hamiltonian, expanding the horizons for long-distance entanglement (LDE) generation in open quantum systems.

Summarizing, our work explored the role of a parity-time ($\hat{\mathcal{P}}\hat{\mathcal{T}}$)-symmetric non-Hermitian Su-Schrieffer-Heeger (SSH) Hamiltonian in the presence of staggered magnetic fields for creating almost frozen LDE between spatially separated quantum links over extended periods of time. Focusing on the $XX$ limit of the SSH model, we analytically derived the exceptional points from the resulting bulk-mediated interactions between the links. We demonstrated that the dynamical states near these exceptional points yield almost unit-entangled quantum links. By exploiting the dynamics in the broken phase, these highly entangled states can be dynamically reached from initially product states, near the exceptional points across different classes of the SSH models. Unlike purely unitary dynamics, where the produced entanglement inherently oscillates over time, the non-Hermitian dynamics are capable of creating robust average entanglement, independent of the initial states. Although high LDE can be achieved for moderate system sizes in the broken-symmetry regime near the exceptional points, the entanglement generally decreases more rapidly with increasing system size for SSH models beyond the (XX) limit. Furthermore, we observed that the applicability of the protocol was not limited to nearest-neighbor interactions, as it consistently generated highly entangled links even in the presence of long-range interactions.


We examined the LDE between the links in the stationary state of the system, corresponding to the lowest energy eigenvalue in the unbroken regime, and to the highest imaginary part of the eigenvalue in the broken regime. We found that for small link-interaction strengths, the stationary states in the unbroken regime possess high LDE, providing an alternative to the dynamical generation of highly entangled links.

Our results reveal a novel route for generating highly entangled quantum links via a non-Hermitian interacting bulk, highlighting the pivotal role of exceptional points in securing robust LDE between spatially separated quantum systems. More broadly, these findings illustrate the potential benefits of reservoir engineering and continuously monitored environments for the preparation of persistent entangled links, opening new avenues for the realization of quantum networks and for extending these ideas in other non-Hermitian architectures.


\acknowledgements

We  acknowledge the use of cluster computing facility at the Harish-Chandra Research Institute. This research was supported in part by the INFOSYS scholarship for senior students. We acknowledge support from the project entitled ``Technology Vertical - Quantum Communication'' under the National Quantum Mission of the Department of Science and Technology (DST)  (Sanction Order No. DST/QTC/NQM/QComm/$2024/2$ (G)).

\appendix

\section{Effective non-Hermitian Hamiltonian}
\label{app:nH_deriv}

Let us consider a spin chain composed of \(N\) spin-\(1/2\) parties, interacting via the Hamiltonian \(\hat{\mathcal{H}}\). For an isolated (closed) system, the evolution is unitary and is entirely governed by the Hermitian Hamiltonian \(\hat{\mathcal{H}}\). In contrast, when each spin is coupled to its own local bath and the initial system–bath state is assumed to be a product state \(\rho(0) = \rho_s(0)\otimes\rho_B(0)\), the dynamics becomes non-unitary and is described by the Gorini–Kossakowski–Lindblad–Sudarshan (GKLS) master equation~\cite{open_quan_book, lidar_2020_lecture}, given by
\begin{eqnarray}
\nonumber {\frac{d\rho_s}{dt}} &=& -i [\hat{\mathcal{H}}, \rho_s] + \sum_{k = 1}^{N} \gamma_k \mathcal{D}[L_{k}]\left( \rho_s\right),\\
\nonumber &=& -i \left(\hat{\mathcal{H}} \rho_s - \rho_s\hat{\mathcal{H}} \right) - \sum_k \frac{\gamma_k}{2}\left(L_k^{\dagger}L_k\rho_s + \rho_sL_k^{\dagger}L_k\right) \\ &&+ \sum_k\gamma_k L_k \rho_s L_k^{\dagger}
\label{eq:mastereqn}
\end{eqnarray}
where \(L_k\) are the site-dependent Lindblad operators at site \(k\) and \(\gamma_k\) are the corresponding dissipation rates for the dissipator $\mathcal{D}[A](\rho) = A\rho A^{\dagger}-\frac{1}{2}\left(A^{\dagger}A\rho + \rho A^{\dagger}A\right)$. Continuously monitoring the baths leads to unraveling of quantum trajectories of the GKLS equation, with \(L_k \rho_s L_k^{\dagger}\) corresponding to quantum jumps~\cite{open_quan_book}. Under post selection of a quantum trajectory, with no jumps, the jump operator term in Eq.~(\ref{eq:mastereqn}) vanishes and the effective dynamics is given by
\begin{eqnarray}
\nonumber {\dot\rho_s(t)} &=& -i [\hat{\mathcal{H}}_{\text{eff}} (t), \rho(t)], \quad
\\ \hat{\mathcal{H}}_{\text{eff}} &=& \hat{\mathcal{H}} - i\sum_{k=1}^{N}\frac{\gamma_k}{2}L_k^{\dagger}L_k 
\end{eqnarray}
where \(\hat{\mathcal{H}}_{\text{eff}}\) is the effective non-Hermitian Hamiltonian that provides a closed-system non-Hermitian representation of the open-system dynamics generated by the Hermitian model~\cite{Minganti2020, turkeshi_prb_2021, turkeshi_prb_2023b}.

In the case of dephasing, with site-dependent Lindblad operators, \(L_k = \frac{1 + (-1)^k \hat{\sigma}_k^z}{2}\) at \(k\)-th site, the effective non-Hermitian Hamiltonian \(\hat{\mathcal{H}}_{\text{eff}}\) is given (upto constant shift with $i\sum_{k=1}^{N}\gamma_k/2$) by
\begin{eqnarray}
\hat{\mathcal{H}}_{\text{eff}} &=& \hat{\mathcal{H}} - i\sum_{k=1}^{N}\frac{\gamma_k}{2}(-1)^k\hat{\sigma}_k^z,
\end{eqnarray}
resulting in alternating magnetic fields. For a 
$\hat{\mathcal{P}}\hat{\mathcal{T}}$-symmetric system $\hat{\mathcal{H}}$, the non-Hermitian Hamiltonian $\hat{\mathcal{H}}_{\text{eff}}$ is also $\hat{\mathcal{P}}\hat{\mathcal{T}}$-symmetric for even $N$.

In our work, we study the role of non-Hermiticity in generation long-distance entanglement between two link spins, \((\ell_1,\ell_2)\) connected by a chain of \(N-2\) intermediate bulk spins, with interactions given by the Hamiltonian \(\mathcal{H}\equiv\hat{\mathcal{H}}_{B\ell}\). Hence, the non-Hermitian effective Hamiltonian \(\hat{\mathcal{H}}_{\text{eff}}\equiv\hat{\mathcal{H}}_{\text{nH}}\) represents the presence of imaginary alternating transverse magnetic field with \(\gamma_k/2 \equiv h'\) as given in Eq.~(\ref{eq:nhH_def}) of the main text.

\section{Logarithmic Negativity}
\label{app:log_neg}

In our work, we aim to compare the link-to-link entangling capabilities of the Hermitian and the non-Hermitian models by incorporating the logarithmic negativity of entanglement as the figure of merit
In order to study the entanglement profile of the links (\(\ell_1,\ell_2\)), we employ the logarithmic negativity \cite{vidal_pra_2002,plenio_prl} as the measure of entanglement, which is defined as
\begin{equation}
   \mathcal{E}_{\ell_1\ell_2} \equiv \mathcal{E}(\rho_{(\ell_1\ell_2)}) \!=\! \log_2[2\mathcal{N}_{(\ell_1\ell_2)}+1], 
\end{equation}
where \(\rho_{(\ell_1\ell_2)}\) is the state of the links (\(\ell_1,\ell_2\)) after tracing out the (\(N-2\)) bulk spins and $\mathcal{N}_{(\ell_1\ell_2)}$ denotes the negativity given by the absolute sum of negative eigenvalues of the partially transposed state over either of (\(\ell_1,\ell_2\)) as $\rho^{T_{\ell_2}}_{(\ell_1\ell_2)}$~\cite{peres_prl_1996, horodecki_pla_1996}.  The logarithmic negativity $\mathcal{E}=0$ provides only a necessary condition of separability for $d_{\ell_1} d_{\ell_2}>6$~\cite{PPTboundPRL, NPTbound00, NPTboundbruss} with $d_{\ell} (= 2 s_{\ell} +1)$ being the dimension of the Hilbert space of individual link spins. Therefore, for the qubits ($d_{\ell_1}=d_{\ell_2}=2$) considered in this work, the logarithmic negativity $\mathcal{E}=0$ is both necessary and sufficient condition, providing a valid measure of entanglement.

Given an arbitrary pure state $\ket{\Psi}\!=\!\sum_{k=1}^{N} c_k\ket{\underline{k}}$ in the $n^{(\uparrow)} \!=\! 1$ sector ($\{\ket{\underline{k}}=\hat{\sigma}^x_k \ket{\downarrow}^{\otimes N}\}_{k=1}^{N}$ are the basis), the entanglement between the links $(\ell_1, \ell_2) \!\equiv\! (1,N)$ is given by
\begin{equation}
    \mathcal{E}_{\ell_1\ell_2} = \log_2\left( \sqrt{\eta_B^2 \!+\! 4|c_1|^2|c_N|^2} + |c_1|^2+|c_N|^2\right),
    \label{eq:ent}
\end{equation}
where $\eta_B=\sum_{k=2}^{N-1}|c_k|^2=1\!-\!|c_1|^2\!-\!|c_N|^2$ is the probability of $\ket{\uparrow}$ in the bulk.

Tracing out the bulk sites, the reduced state \(\rho_{1.N}\) and the partially transposed state \(\rho^{T_{\ell_2}}_{1,N}\) of the links $(\ell_1, \ell_2)$ are given by
\begin{align}
    \rho_{1,N} \!=\!\! \left[\begin{smallmatrix} \eta_B & 0 & 0 & 0 \\ 0 & |c_N|^2 & c_N c_1^* & 0 \\ 0 & c_1 c_N^* & |c_1|^2 & 0 \\ 0 & 0 & 0 & 0 \end{smallmatrix}\right], \ 
    \rho^{T_{\ell_2}}_{1,N} \!=\!\! \left[\begin{smallmatrix} \eta_B & 0 & 0 & c_N c_1^* \\ 0 & |c_N|^2 & 0 & 0 \\ 0 & 0 & |c_1|^2 & 0 \\ c_1 c_N^* & 0 & 0 & 0 \end{smallmatrix}\right]
    \label{eq:mats}
\end{align}
which are in the computational basis $\{\ket{\downarrow\downarrow}, \ket{\downarrow\uparrow}, \ket{\uparrow\downarrow}, \ket{\uparrow\uparrow}\}$. Since $|c_1|^2,|c_N|^2>0$, the negativity $\mathcal{N}_{(1,N)} = \frac{1}{2}\left(\sqrt{\eta_B^2 \!+\! 4|c_1|^2|c_N|^2}-\eta_B\right)$ gives the required measure of entanglement. Using the condition $|c_1|^2+|c_N|^2\leq 1$, $\mathcal{E}_{\ell_1\ell_2}=1$ iff $|c_1|^2=|c_N|^2=\frac{1}{2}$, i.e., when $\rho_{1,N}=\ket{\Phi^{(\tau)}}\bra{\Phi^{(\tau)}}_{1,N}$ with $\ket{\Phi^{(\tau)}}_{1,N}=\frac{1}{\sqrt{2}}\left(\ket{\uparrow\downarrow}_{1,N}+e^{i\tau}\ket{\downarrow\uparrow}_{1,N}\right)$.



\section{Entanglement generation for the $XX$ model without magnetic fields}
\label{app:xx_ent}
In this section, we compare the entanglement generation from the product state, with the distribution of entanglement with the quantum state transfer for the the $XX$ model, i.e., $\delta=\lambda=1, h=0$ in Eq.~(\ref{eq:nhH_def}) of the main text. We show that creating entanglement between the links $(\ell_1, \ell_2)$ at ends $(1,N)$ of the chain from a product state is more stringent than the distribution of the entanglement across the chain. 

In the basis $\{\ket{\underline{k}}\}_{k=1}^{N}$ for $N$ sites, the Hamiltonian of the system $\hat{\mathcal{H}}_{B\ell}(\delta \!=\!\lambda \!=\! 1, h \!=\! 0)=\hat{\mathcal{H}}_{XX}$, the eigenvalue expression and the eigenvectors after its diagonalization are given by~\cite{Leib_1961, Bose_2003}
\begin{align}
    &\hat{\mathcal{H}}_{XX} \!=\! 2\sum_{k=1}^{N-1}\ket{\underline{k}}\!\bra{\underline{k\!+\!1}} + \ket{\underline{k\!+\!1}}\!\bra{\underline{k}} = \sum_{\mathclap{a=1}}^{N} E_a \!|\tilde{a}\rangle\!\langle\tilde{a}|; \\
    &E_a \!=\! 4\cos \!\left(\!\frac{a\pi}{N\!+\!1}\!\right); \quad  |\tilde{a}\rangle \!=\! \sqrt{\!\frac{2}{N\!+\!1}}\sum_{\mathclap{b=1}}^{N} \! \sin \!\left(\!\frac{ab\pi}{N\!+\!1}\!\right) \! \ket{\underline{b}}. \nonumber
\end{align}
Therefore, starting with an initial state $\ket{\Psi(0)} = \ket{\underline{1}}$, the maximum value $|c_N(t)|=|\langle\underline{N}|e^{-i\hat{\mathcal{H}}_{XX}t}|\underline{1}\rangle|$ (see Eq. \ref{eq:mats}) scales as $|c_N(t^*_N)|\sim N^{-1/3}$ (see Eq. (12) of~\cite{Bose_2003}), at optimal time $t^*_N$, with system-size $N$. As the high entanglement $\mathcal{E}_{\ell_1 \ell_2}\approx 1$ between the links $(\ell_1, \ell_2)$ requires $|c_1|^2\approx |c_N|^2 \approx \frac{1}{2}$, the links cannot be maximally entangled via $XX$ chain, and
\begin{equation}
    \max_t\mathcal{E}_{\ell_1 \ell_2}(t)\sim \frac{3|c_1(t^*_N)|^2}{\log 2}N^{-2/3} \approx4.328|c_1(t^*_N)|^2N^{-2/3}
\end{equation}
for large $N$ and strong interaction strength $\lambda=1$. Note that, $|c_1(t^*_N)|$ can also be \(N\) dependent, which can also decrease with system-size. Therefore, the maximum entanglement between the links decreases worse than $N^{-2/3}$, whereas Ref.~\cite{Bose_2003} shows that the maximum entanglement in quantum state transfer decreases as $N^{-1/3}$. Therefore, creating entanglement at large distances from a product state can be more demanding than transferring a part of an entangled state.

\section{Resultant Hamiltonian of the links for the $XX$ bulk with small link strengths and magnetic fields}
\label{app:xx_sw}

The \(\hat{\mathcal{P}}\hat{\mathcal{T}}\)-symmetric non-Hermitian SSH Hamiltonian $\hat{\mathcal{H}}_{\text{nH}}$ in Eq.~(\ref{eq:nhH_def}) of the main text can be solved perturbatively for $\delta=1$, i.e., the bulk with $XX$ interactions, for weak link strengths and small magnetic fields $\lambda,h\ll 1$. The eigenspectrum of the unperturbed $\hat{H}_B\equiv\hat{H}_{XX}$ and the perturbation $\hat{V}_{\ell}$, given in the Eqs.~(\ref{eq:xx_spec}) and~(\ref{eq:vw}) of the main text are,
\begin{align}
    &\hat{\mathcal{H}}_{\text{nH}} (\delta\!=\!1) = \hat{H}_{XX} \!+\! \lambda\hat{V}_{\ell} \!+\! ih\hat{W}; \  \nonumber\\
    &\hat{H}_{XX} \!=\! \sum_{a=2}^{N-1} \!E_{a}\ket{\tilde{a}}\!\!\bra{\tilde{a}};\quad E_{a} \!=\! 4\cos\theta_a, \quad \theta_a \!\equiv\! \frac{(a\!-\!1)\pi}{N\!-\!1}, \nonumber\\
    &\ket{\tilde{a}} = \sqrt{\frac{2}{N\!-\!1}}\sum_{b=2}^{N-1} 
    \sin\left[(b\!-\!1)\theta_a\right] |\underline{b}\rangle,\quad (2 \!\leq\! a \!\leq\! N\!-\!1) \nonumber \\
    & E_1 \!=\! E_N \!=\! 0, \quad|\tilde{1}\rangle \!=\! |\underline{1}\rangle, |\tilde{N}\rangle \!=\! \ket{\underline{N}}, \quad P_0 \!=\! |\tilde{1}\rangle\!\langle\tilde{1}| \!+\! |\tilde{N}\rangle\!\langle\tilde{N}| \nonumber \\
    &\hat{V}_{\ell} =  \sum_{a=2}^{N-1} 2v_a \Big[ |\tilde{1}\rangle\!\langle\tilde{a}| + |\tilde{a}\rangle\!\langle\tilde{1}| + (-1)^{a}(|\tilde{N}\rangle\!\langle\tilde{a}| + |\tilde{a}\rangle\!\langle\tilde{N}|) \Big], \nonumber\\
    &\hat{W} = 2\left(|\tilde{1}\rangle\langle\tilde{1}| - |\tilde{N}\rangle\langle\tilde{N}|+\sum_{a=2}^{N-1}\ket{-\tilde{a}}\bra{\tilde{a}}\right), 
\end{align}
in the eigenbasis of the $\hat{H}_{XX}$ with $v_a\!=\!\sqrt{\frac{2}{N-1}}\sin\theta_a$ and $-a\equiv N\!+\!1\!-\!a$. As the system-size of the bulk ($N-2$) as even, there are no zero energy states in the bulk. Therefore, $P_0$ is the zero energy subspace of interest as the links $(\ell_1, \ell_2)\equiv(1,N)$ are uncoupled from the bulk, with $P_B \!=\! \mathbb{I} \!-\! P_0$ as the projector on the bulk states. Hence, only $\hat{V}_{\ell}$ couples $P_0$ and $P_B$ and gives the off-diagonal terms in the $P_0$ subspace, while $\hat{H}_{B}$ and $\hat{W}$ does not, i.e.,
\begin{equation}
    \hat{V}_{\ell} \!=\! P_0 \hat{V}_{\ell} P_B \!+\! P_B \hat{V}_{\ell} P_0, \quad \hat{A} \!=\! P_0 \hat{A} P_0 \!+\! P_B \hat{A} P_B, \nonumber
\end{equation}
with $\hat{A} \!=\! \hat{H}_{XX} \!+\! ih\hat{W}$.

Using Schrieffer-Wolff transformation~\cite{Bravyi_2011, Starkov_2023} for $\lambda\ll 1$, the appropriate generator $\hat{S}$ gives the transformed Hamiltonian of the links as
\begin{align}
    & \hat{\mathcal{H}}_{\text{nH}}^{(S)} \!=\! e^{\hat{S}} \hat{\mathcal{H}}_{\text{nH}} e^{-\hat{S}} \!=\! \hat{\mathcal{H}}_{\text{nH}} + [\hat{S}, \hat{\mathcal{H}}_{\text{nH}}] + \frac{1}{2}[\hat{S}, [\hat{S}, \hat{\mathcal{H}}_{\text{nH}}]] + O(\hat{S}^3), \nonumber \\
    & \hat{H}_{\ell_1 \ell_2}^{\text{SW}} = P_0 \hat{\mathcal{H}}_{\text{nH}}^{(S)} P_0^{\dagger}. \nonumber
\end{align}
Since, $\hat{W}$ is diagonal in the subspace of interest i.e., $P_0$, there are no off-diagonal terms to vanish and $\hat{S}\propto \lambda$. Hence, keeping $[S, \hat{H}_{XX}] = -\lambda\hat{V}_{\ell}$ gives
\begin{equation}
    \hat{\mathcal{H}}_{\text{nH}}^{(S)} = \hat{H}_{XX} + ih\hat{W} + ih[\hat{S}, \hat{W}] + \frac{\lambda}{2}[\hat{S}, \hat{V}_{\ell}] +O(\lambda^3), \nonumber
\end{equation}
and the first-order terms of $\lambda$ vanishes ($\hat{S}\propto \lambda$) in the off-diagonal in $P_0$ subspace. The generator can be expressed in the eigenbasis $\{\tilde{k}\}_{k=1}^N$ of $\hat{H}_{XX}$ from the elements of $\hat{V}_{\ell}$ in Eq.~(\ref{eq:vw}).
\begin{widetext}
\begin{align}
    &\langle\tilde{a}|[\hat{S}, \hat{H}_{XX}]|\tilde{b}\rangle \!=\! \langle\tilde{a}|\hat{S}|\tilde{b}\rangle E_b \!-\! E_{a} \langle\tilde{a}|\hat{S}|\tilde{b}\rangle \!=\! -\lambda \langle\tilde{a}|\hat{V}_{\ell}|\tilde{b}\rangle \implies \hat{S} \!=\! \frac{\lambda}{2} \sum_{\mathclap{a=2}}^{\mathclap{N-1}} \frac{v_a}{\cos\theta_a} \left[ |\tilde{a}\rangle\!\langle\tilde{1}| \!-\! |\tilde{1}\rangle\!\langle\tilde{a}| \!+\! (-1)^a\left(|\tilde{a}\rangle\!\langle\tilde{N}| \!-\! |\tilde{N}\rangle\!\langle\tilde{a}| \right)\right], \nonumber\\
    & \langle\tilde{1}|[\hat{S}, \hat{V}_{\ell}]|\tilde{1}\rangle \!=\! \langle\tilde{N}|[\hat{S}, \hat{V}_{\ell}]|\tilde{N}\rangle \!=\! \frac{-4\lambda}{N\!-\!1} \!\sum_{\mathclap{a=2}}^{\mathclap{N-1}}\frac{\sin^2\theta_a}{\cos\theta_a} \!=\! 0, \  \nonumber \langle\tilde{1}|[\hat{S}, \hat{V}_{\ell}]|\tilde{N}\rangle \!=\! \langle\tilde{N}|[\hat{S}, \hat{V}_{\ell}]|\tilde{1}\rangle = \frac{-4\lambda}{N\!-\!1} \!\sum_{\mathclap{a=2}}^{\mathclap{N-1}}(-1)^a \frac{\sin^2\theta_a}{\cos\theta_a}.
\end{align}
\end{widetext}
Similarly $P_0[\hat{S}, \hat{\tilde{W}}]P_0^{\dagger} \!=\! 0$. Hence, in the original basis $|\tilde{1}\rangle \!=\! |\underline{1}\rangle, |\tilde{N}\rangle \!=\! \ket{\underline{N}}$
\begin{equation}
    \hat{H}_{\ell_1 \ell_2}^{\text{SW}} \!=\!  
    \begin{pmatrix} 
        2ih & \frac{-2\lambda^2}{N\!-\!1}\chi_{\text{e}} \\ 
        \frac{-2\lambda^2}{N\!-\!1}\chi_{\text{e}} & -2ih 
    \end{pmatrix}, \quad  \chi_{\text{e}} = \sum_{\mathclap{a=2}}^{\mathclap{N-1}} (-1)^{a} \frac{\sin^2\theta_a}{\cos \theta_a}, \nonumber
\end{equation}
is the resulting interaction between the links with $\lambda\ll1$. The value of $\chi_{\text{e}}=-(-1)^{N/2}(N\!-\!1)$ is calculated and validated numerically, which results in Eq.~(\ref{eq:H_SW}) of the main text. Note that the resultant Hamiltonian for the bulk states, i.e., in $P_B$ subspace will contain the $h^2$ and the $h\lambda$ terms as well, as both $h$ and $\lambda$ are in the off-diagonals. The eigenvalues of the bulk states are computed by non-degenerate perturbation theory.

\begin{figure}
    \centering
\includegraphics[width=1.0\linewidth]{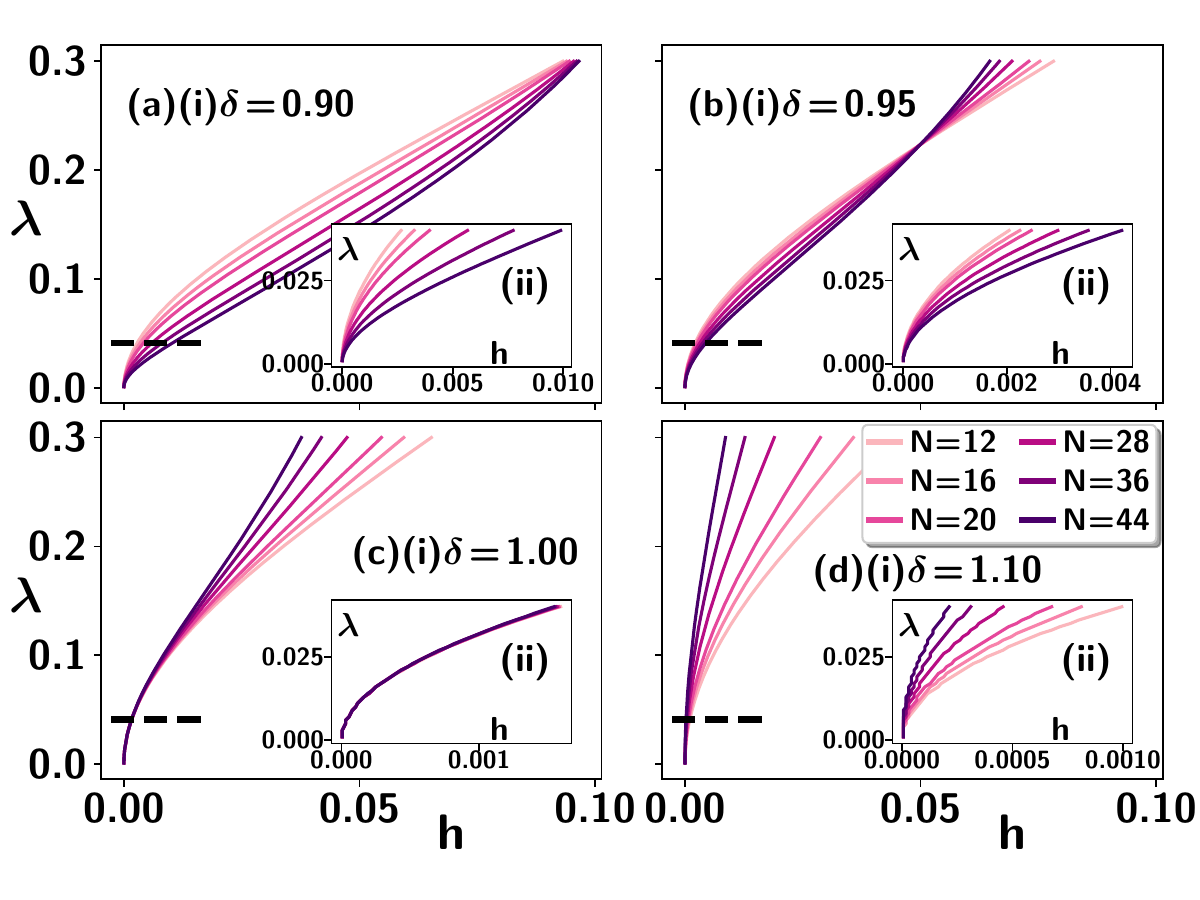}
    \caption{ {\bf Exceptional points} of the non-Hermitian SSH model with link-interaction strength $\lambda$ (ordinate) and staggered magnetic field $h$ (abscissa) for different system-sizes $N$ and (a)-(d) different intracell hopping $\delta$ of the bulk. The black dashed line denotes $\lambda_c=0.04$, where each inset show the exceptional points for $\lambda<\lambda_c$, with weak link-interaction strengths. All the axes are dimensionless.}
    \label{fig:ssh_eps}
\end{figure}

\section{Exceptional points in the SSH model for weak link-interaction strengths}
\label{app:ssh_ep}

\begin{figure}
\includegraphics[width=1.0\linewidth]{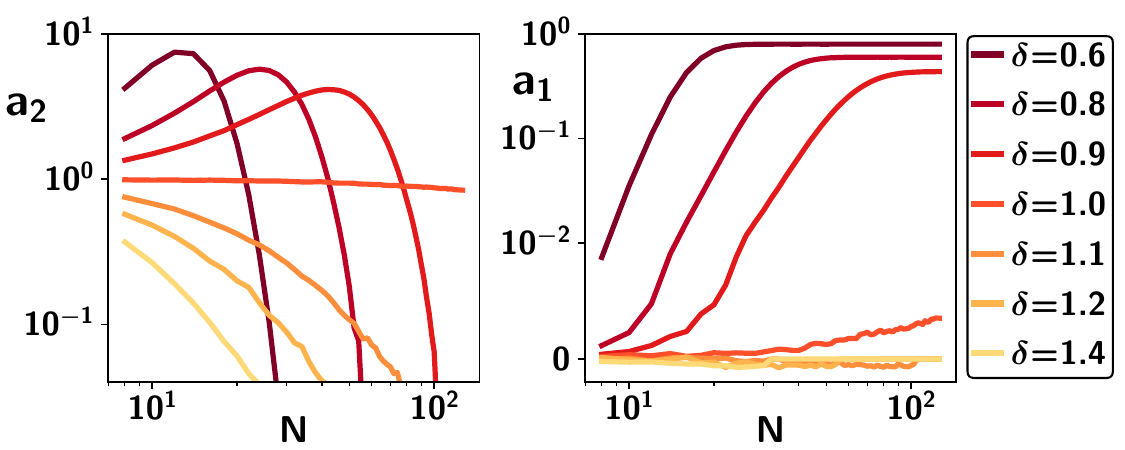}
    \caption{ The coefficients of exceptional points of the non-Hermitian SSH model with link-interaction strength with the system-size $N$ (abscissa). The exceptional points are $h_e=a_2\lambda^2 +a_1\lambda$ with (a) $a_2$ and (b) $a_1$ on the ordinates. Different lines denote different intracell hopping $\delta$ of the bulk. All the axes are dimensionless.}
    \label{fig:ssh_eps_coeffs}
\end{figure}

The exceptional points of the $\hat{\mathcal{P}}\hat{\mathcal{T}}$-symmetric non-Hermitian SSH model $\hat{\mathcal{H}}_{\text{nH}}$ (see Eq.~\ref{eq:nhH_def} of the main text) with weak link strengths $\lambda\ll1$ are determined by numerically diagonalizing the Hamiltonian $\hat{\mathcal{H}}_{\text{nH}}$ in the $n^{(\uparrow)} \!=\! 1$ sector. From the eigenvalues $\{E_k\}_{k=1}^{N}$, we compute the maximum imaginary component of the eigenvalue as $e_i=\max_{k=1}^N \text{imag}(E_k)$, where $\text{imag}(E_k)$ denotes the imaginary component of the eigenvalue $E_k$. For a given $\delta$, system-size $N$ and link-interaction strength $\lambda$, we numerically find $h_e$ such that $e_i(h_{e}\!-\!dh/2)>\epsilon$ and $e_i(h_{e}\!+\!dh/2)<\epsilon$, for small $\epsilon=10^{-9}$ and $dh=10^{-4}$.

As illustrated in Fig.~\ref{fig:ssh_eps}, the exceptional points $h_e(\lambda)$ vary with the system size $N$ and the intracell interaction $\delta$ of the SSH bulk. For weak link-interaction strengths $\lambda < \lambda_c = 0.04 \ll 1$, we obtain $h_e = \lambda^2$ for $\delta = 1$ across all system sizes $N$, a result analytically demonstrated in Theorem~\ref{thm:xx_ep} of the main text. For $\delta \neq 1$, the exceptional lines within the regime $\lambda < \lambda_c$ are numerically fitted using the quadratic expression $h_e = a_2\lambda^2 + a_1\lambda + a_0$. While the intercept $a_0 \sim 10^{-6}$ remains consistently small for all considered values of $\delta$ and $N$, the coefficients $a_2$ and $a_1$ depend on both the system size and $\delta$, as depicted in Fig.~\ref{fig:ssh_eps_coeffs}. 

Specifically, for $\delta > 1$, $a_2$ decreases with increasing system size $N$ while $a_1 \sim 10^{-4}$. This indicates for large system-sizes, the non-Hermitian SSH model with $\delta > 1$ and weak link interactions $\lambda \ll 1$ resides in the broken phase for any $h > 0^+$. Conversely, for $\delta < 1$, $a_2$ initially increases with $N$ before decreasing for sufficiently large system sizes. Meanwhile, the coefficient $a_1$ increases with $N$ and saturates to a fixed value in the large-$N$ limit. Consequently, for $\delta < 1$ and large system sizes, the exceptional points of the non-Hermitian SSH model with weak link interactions $\lambda \ll 1$ scale linearly as $h_e \approx a_1\lambda$.

\bibliography{refbulk.bib}
\end{document}